\documentclass{LMCS} 

\def\doi{8(4:5)2012}
\lmcsheading%
{\doi}
{1--25}
{}
{}
{Dec.\phantom.06, 2009}
{Oct.~10, 2012}
{}

\newcommand{\bull}{\rule{.85ex}{1ex} \par \bigskip}

\newcommand{\ignore}[1]{}

\usepackage{cite}
\usepackage[latin1]{inputenc}
\usepackage{amsmath,amsxtra,amsfonts,amssymb}
\usepackage{graphicx}
\usepackage{cite}
\usepackage{xspace, varioref}
\usepackage[matrix, curve, arrow, tips, frame]{xy}
\usepackage{enumerate,hyperref}
\newcommand{\cproblem}[3]{
\vspace{.2cm}
\noindent {\bf #1} \\
INPUT: #2 \\
QUESTION: #3 \\}

\DeclareMathOperator{\Csp}{CSP}
\DeclareMathOperator{\GLP}{GLP}

\newcommand{\R} {\mathbb{R}}

\newcommand{\lingamma}{\Gamma_{{\rm lin}}}

\begin{document}

\title{Essential Convexity and Complexity of Semi-Algebraic Constraints}

\author[M~.Bodirsky]{Manuel Bodirsky\rsuper a}
\address{{\lsuper a}CNRS/LIX, \'Ecole Polytechnique, 91128 Palaiseau,
  France}
\email{bodirsky@lix.polytechnique.fr}
\thanks{{\lsuper a}Manuel Bodirsky has received funding from the \emph{European Research Council} 
under the European Community's Seventh Framework Programme (FP7/2007-2013 Grant Agreement no. 257039).}

\author[P.~Jonsson]{Peter Jonsson\rsuper b}
\address{{\lsuper b}Department of Computer and System Science, Link\"opings Universitet\\
 SE-581 83, Sweden.} 
 \email{petej@ida.liu.se}
\thanks{{\lsuper b}Peter Jonsson is partially supported by the {\em Center for Industrial Information Technology}
({\sc Ceniit}) under grant 04.01 and by the {\em Swedish Research Council} (VR) under grants 2006-4532 and 621-2009-4431.}

\author[T.~v.~Oertzen]{Timo von Oertzen\rsuper c}
\address{{\lsuper c}Max-Planck-Institute for Human Development, K\"onigin-Luise-Strasse 5, 14195 Berlin, Germany, and
University of Virginia, Department of Psychology, Charlottesville, USA.}
\email{vonoertzen@mpib-berlin.mpg.de}


\maketitle

\begin{abstract}
Let $\Gamma$ be a structure with a finite relational signature and a first-order definition in $(\mathbb R;*,+)$ with parameters from $\mathbb R$, that is, a relational structure over the real numbers where all relations are semi-algebraic sets.
In this article, we study the computational complexity of
constraint satisfaction problem (CSP) for $\Gamma$: the problem
to decide whether a given primitive positive sentence
is true in $\Gamma$.
We focus on those structures $\Gamma$ that 
contain the relations $\leq$, 
$\{(x,y,z) \; | \; x+y=z\}$ and $\{1\}$. Hence, all CSPs studied in this article are at least as expressive as the feasibility problem for linear programs.
The central concept in our investigation is {\em essential convexity}: a
relation $S$ is essentially convex if
for all $a,b \in S$, there are only finitely many
points on the line segment between $a$ and $b$ that are not in $S$.
If $\Gamma$ contains a relation $S$ that is not essentially convex
and this is witnessed by \emph{rational} points $a,b$, 
then we show that the CSP for $\Gamma$ is NP-hard.
Furthermore, we characterize essentially convex relations in
logical terms. This different view may
open up new ways for identifying tractable classes of semi-algebraic CSPs.
For instance, we show that 
if $\Gamma$ is a first-order expansion of $(\mathbb R; +,1,\leq)$, then the CSP for $\Gamma$ can be solved in polynomial time if and only if all relations in $\Gamma$ are essentially convex (unless P=NP).
\end{abstract}

\section{Introduction}
Linear Programming is a computational problem of outstanding theoretical
and practical importance. It is known to be computationally
equivalent to the problem to decide whether a given set
of linear (non-strict) inequalities is \emph{feasible}, i.e., defines a non-empty set:

\cproblem{Linear Program Feasibility}
{A finite set of variables $V$; a finite set of linear inequalities of the form $a_1x_1+ \cdots+a_k x_k \leq a_0$ where $x_1,\dots,x_k \in V$ and $a_0,\dots,a_k$ are
rational numbers where the numerators and denominators 
are represented in binary.}
{Does there exist an $x \in {\mathbb R}^{|V|}$ that satisfies all inequalities?}

This problem can be viewed as a \emph{constraint
satisfaction problem}, where the allowed constraints are linear
inequalities with rational coefficients, and the question is whether there is an assignment
of real values to the variables such that all the constraints are satisfied. 
For formal definitions of concepts related to constraint satisfaction, we refer
the reader to Section~\ref{sect:csp}.
It is obvious that this problem cannot be formulated
with a finite constraint language; 
however, we will later on (Proposition~\ref{prop:lp-pp}) see that
the feasbility problem for linear programs
is polynomial-time equivalent to the constraint satisfaction problem
for the structure 
$$\lingamma := \big ({\mathbb R}; \{(x,y,z) \; | \; x+y=z\},\leq,\{1\} \big)\; .$$

It is well-known that linear programming can be solved in polynomial time; moreover, several algorithms are known that are efficient also in practice. 
In this article, we study how far $\lingamma$ can be expanded such that the corresponding
constraint satisfaction problem remains polynomial-time solvable.
An important class of relations that generalizes the class
of relations defined by linear inequalities is the class of all
\emph{semi-algebraic} relations, i.e., relations 
that have a first-order definition over 
$(\mathbb R; *,+)$ using parameters from $\mathbb R$. By the fundamental theorem of Tarski and Seidenberg, it is known that a relation $S \subseteq {\mathbb R}^n$
is semi-algebraic if and only if it has a quantifier-free first-order definition in $(\mathbb R; *,+,\leq)$ using parameters from $\mathbb R$. Geometrically, we 
can view semi-algebraic sets as finite unions of finite intersections of the solution
sets of strict and non-strict polynomial inequalities.

We propose
a framework for systematically studying the computational
complexity of expansions of $\lingamma$ by semi-algebraic relations.
In this framework, a constraint satisfaction problem is given by
a (fixed and finite) constraint language $\Gamma$. All the constraints
in the input of such a feasibility problem must be chosen from this 
constraint language $\Gamma$ (a formal definition can be found in Section~\ref{sect:csp}).
This way of parameterizing constraint satisfaction problems by their
constraint language has proved to be very fruitful for finite domain
constraint satisfaction problems~\cite{FederVardi,JBK,BoundedWidth,Conservative,Bulatov}. 
Since the constraint
language is finite, the computational complexity of such a problem
does not depend on how the constraints
are represented in the input. We believe that the very
same approach is very promising for studying the complexity
of problems in real algebraic geometry. In Section~\ref{sect:concl}
we will discuss a connection between some of the CSPs with
semi-algebraic constraint languages and open problems
in convex geometry and semidefinite programming.

One of the key reasons why linear program feasibility can be decided in polynomial time is that the feasible regions of linear inequalities are \emph{convex}.
Convexity is not a necessary condition for tractability of 
semi-algebraic constraint satisfaction problems, though. It is, for instance, well-known
that linear program feasibility can also be decided in polynomial time
when some of the input constraints are \emph{disequalities},
i.e., constraints of the form $a_1 x_1 + \cdots + a_k x_k \neq a_0$
for rational values $a_0,\dots,a_k$.
However, we show that if $\lingamma \subseteq \Gamma$ and $\Gamma$
contains a relation $S$ with rational $a,b \in S$ such that on the line segment $L$ between $a$ and $b$
there are infinitely many points that are not in $S$, then
the $\Csp(\Gamma)$ is NP-hard. This motivates the notion of \emph{essential convexity}: a set $S \subseteq {\mathbb R}^k$ is \emph{essentially convex} if for all
$p,q \in S$ there are only finitely many points on the line between $p$ and $q$ that are not in $S$.
One of our central results is a logical characterization
of essentially convex semi-algebraic relations in Section~\ref{sect:main}. This characterization can be used to show several results
that are briefly described next.

A relation is called \emph{semi-linear}
if it has a first-order definition with rational parameters\footnote{We deviate from model-theoretic terminology as it is used e.g. in~\cite{MarkerPeterzilPillay} in that we only allow rational and not arbitrary real parameters in first-order definitions. Our definition conincides with the definition of semi-linear sets given in e.g.~\cite{DumortierGyssensVardeurzenVanGucht,
DumortierGyssensVardeurzenVanGuchtErr}.}
 in the structure $(\mathbb R; +,\leq)$.
From the perspective of constraint satisfaction, 
the set of semi-linear relations is a rich set.
For example, every relation $S \subseteq {\mathbb Q}^k$ with finitely 
many elements is semi-linear; thus, every finitary
relation on a finite set can be viewed as a semi-linear relation.
In Section~\ref{sect:semilinear}, we show that 
when we add a finite number of semi-linear relations to $\lingamma$, then the resulting
language either has a polynomial-time or an NP-hard constraint satisfaction problem.
This result is useful for studying optimization problems: note that
linear programming can be viewed as optimizing a linear function over
the feasible points of a set of linear inequalities. This
view suggests an immediate generalization: optimize a linear function over
the feasible points of an instance of a constraint satisfaction problem
for semi-linear constraint languages.
We completely classify the complexity of this problem in Section~\ref{sect:glp}.

Another application concerns temporal reasoning.
A {\em temporal constraint language} $\Gamma$ is a structure
$({\mathbb R};R_1,\ldots,R_l)$ with a first-order definition
in $({\mathbb R};<)$. Many computational problems in 
artificial intelligence and scheduling can be modeled 
as constraint satisfaction problems for temporal constraint languages.
The complexity of the CSP for temporal constraint languages $\Gamma$ has been completely classified recently~\cite{tcsps-journal}; there are
9 tractable classes of temporal constraint satisfaction problems.
Often, temporal languages are extended with some
mechanism for expressing {\em metric} time, i.e., the ability to
assign numerical values to variables and performing some kind of
arithmetic calculations~\cite{DrakengrenJonssonMetric}.
It has been observed that many metric languages $\Gamma$ are semi-linear
and satisfy $\lingamma \subseteq \Gamma$, and if such
a language is
polynomial-time solvable, then it is a subclass
of the so-called Horn-DLR class~\cite{JonssonBaeckstroem}. Our result shows that this is
not a coincidence: whenever $\Gamma$ is not
a subclass of Horn-DLR, then the $\Csp(\Gamma)$ is NP-hard.

\section{Preliminaries} 
\subsection{Constraint Satisfaction Problems}\label{sect:csp}
A first-order formula\footnote{Our terminology is standard; all notions that are not explicitly introduced can be found in standard textbooks, e.g., in~\cite{Hodges}.} is called \emph{primitive positive} (pp)
if it is of the form 
\[\exists x_1,\dots,x_n . (\psi_1 \wedge \dots \wedge \psi_m)\]
where $\psi_i$ are atomic formulas, i.e., formulas of the form
$x=y$ or $S(x_{i_1},\dots,x_{i_k})$ where $S$ is the relation
symbol for a $k$-ary relation in $\Gamma$. We call such
a formula a {\em pp-formula}, and as usual a pp-formula without free variables is called a pp-sentence.

Let $\Gamma=(D; S_1,\dots,S_l)$ be a structure with domain $D$ and a finite relational signature.
The \emph{constraint satisfaction problem for $\Gamma$} 
($\Csp(\Gamma)$ in short)
is the computational problem to decide whether a given 
primitive positive sentence $\Phi$ involving relation
symbols for the relations in $\Gamma$ is true in $\Gamma$.
The conjuncts in a pp-sentence $\Phi$ are also called the
\emph{constraints} of $\Phi$, and
to emphasize the connection between the structure $\Gamma$
and the constraint satisfaction problem, we typically refer to $\Gamma$ as
a \emph{constraint language}.
By choosing an appropriate constraint language $\Gamma$,
many computational problems that have been studied in the literature can be formulated as $\Csp(\Gamma)$
(see e.g.~\cite{BodirskySurvey,JBK}).

When studying the complexity of different CSPs, it is often useful
to be able to derive new relations from old.
If $\Gamma=(D;S_1,\dots,S_l)$ is a relational structure 
and $S \subseteq D^k$ is a relation, then $(\Gamma,S)$ denotes
the expansion $(D;S,S_1,\dots,S_l)$
of the structure $\Gamma$ by the relation $S$.
We say that an $n$-ary relation $S$ is 
{\em pp-definable} in $\Gamma$ if there exists a pp-formula $\phi$
with free variables $x_1,\dots,x_n$ such that $(x_1,\dots,x_n) \in S$
iff $\phi(x_1,\dots,x_n)$ holds in $\Gamma$.
The following simple but important result explains the importance
of pp-definability for constraint satisfaction problems.

\begin{lem}[Jeavons et al.~\cite{JeavonsClosure}] \label{lem:pp}
Let $\Gamma$ 
be a relational structure, and let $S$
be pp-definable over $\Gamma$. Then $\Csp((\Gamma,S))$ is polynomial-time equivalent to $\Csp(\Gamma)$.
\end{lem}

\subsection{Semi-algebraic and semi-linear relations}
\label{sect:prelims}
We say that a relation $S \subseteq D^n$ is \emph{first-order definable} in a structure
$\Gamma$ with domain $D$ if there exists a formula $\phi(x_1,\dots,x_n)$ using universal and existential quantification, disjunction, conjunction, negation, and atomic formulas over $\Gamma$ (where $x_1,\dots,x_n$ denote the free variables in $\phi$) such that $\phi(a_1,\dots,a_n)$ is true
over $\Gamma$ if and only if $(a_1,\dots,a_n) \in S$.
We always admit equality when building atomic formulas, i.e., we
have atomic formulas of the form $t_1 = t_2$ for terms $t_1,t_2$ formed from function symbols for $\Gamma$ and variables.
We say that $S$ is \emph{first-order definable in $\Gamma$ 
with parameters from $A$}, for $A \subseteq D$, if additionally we can use constant symbols
for the elements of $A$ in the first-order definition of $S$.

A set $S \subseteq \mathbb R^n$ is called \emph{semi-algebraic} if it has a first-order definition in $(\mathbb R; *,+)$ using parameters from $\mathbb R$.
Note that the order $\leq$ of the real numbers is first-order definable 
in $(\mathbb R; *,+)$, since $$a \leq b \; \Leftrightarrow \; \exists c. \; b=a+c*c \; .$$

We need some basic algebraic and topological concepts and facts.

\begin{defi}[Section 3.1 in~\cite{BasuPollackRoy}]
A set $S \subseteq {\mathbb R}^n$ is \emph{open} 
if it is the union of open balls, i.e., if every point
of $S$ is contained in an open ball contained in $S$. 
A set $S \subseteq {\mathbb R}^n$ is \emph{closed} if
its complement is open. The \emph{closure} of a set $S$,
denoted $\bar S$, 
is the intersection of all closed sets containing $S$. 
Equivalently, $\bar S = \{ x \in {\mathbb R}^n \; | \; \forall r>0 \; \exists y \in S. \; (y-x)^2 < r^2 \}$.
A point $p$ in $S$
is a {\em boundary point} if for every $\epsilon > 0$, the $n$-dimensional open
ball with radius $\epsilon$ around $p$ contains at least one point in $S$ and
one point not in $S$. The set of boundary points is denoted $\partial S$.
The \emph{interior} of $S$, denoted
by $S^\circ$, is $S \setminus \partial S$. 
\end{defi}
Note that 
the interior of $S$ consists of all $p \in S$ such that
there exists an $\epsilon > 0$ with the following property: 
the $n$-dimensional open ball
with radius $\epsilon$ around $p$ is contained in $S$.
Also note that a finite union of closed sets is closed. 

\begin{prop}[Proposition~2.2.2. in~\cite{BochnakCosteRoy}]
\label{prop:semi-algebraic-closure}
The closure of a semi-algebraic relation is semi-algebraic.
\end{prop}

We use the notion of \emph{dimension} $\dim(S) \in \mathbb N$ of a semi-algebraic set $S$
as defined in~\cite{BochnakCosteRoy}. 

\begin{defi}[Section 2.8 in~\cite{BochnakCosteRoy}]
Let $S \subseteq \mathbb R^k$ be a semi-algebraic set, and let
${\mathcal P}(S)$ be the ring of polynomial functions on $S$, i.e., 
the ring of functions $S \rightarrow {\mathbb R}$ which are the restriction of a polynomial. Then
the dimension of $S$, denoted by $\dim(S)$, is the maximal length
of chains of prime ideals of ${\mathcal P}(S)$, i.e., the maximal $d$
such that there exist distinct prime ideals $I_0, I_1,\dots,I_d$ of ${\mathcal P}(S)$ with $I_0 \subset I_2 \subset \dots \subset I_d$.
\end{defi}

To work with this definition of dimension, we need some more concepts. 

\begin{defi}[see~\cite{BochnakCosteRoy}]
Let $S \subseteq \mathbb R^k$ and $T \subseteq \mathbb R^l$
be semi-algebraic sets. A function $f \colon S \rightarrow T$ is \emph{semi-algebraic} if the set $\{(x_1,\dots,x_k,y_1,\dots,y_l) \; | \; f(x_1,\dots,x_k)=(y_1,\dots,y_l) \}$ is a semi-algebraic
subset of $\mathbb R^{k+l}$. 
\end{defi}

As usual, bijective functions $f \colon S \rightarrow T$ such that 
$S' \subseteq S$ is open if and only if $f(S') \subseteq T$ 
is open are called \emph{homeomorphisms}.



\begin{lem}[Propositions 2.8.5, 2.8.9, and 2.8.13 in~\cite{BochnakCosteRoy}]\label{lem:dim}
Let $S \subseteq \mathbb R^n$ be semi-algebraic. 
\begin{iteMize}{$\bullet$}
\item If $S=S_1 \cup S_2$ then $\dim(S) = \max(\dim(S_1),\dim(S_2))$.
\item If there is a semi-algebraic homeomorphism from $S$ to $(0,1)^d$, then $\dim(S)=d$.
\item $\dim(\bar S \setminus S) < \dim(S)$.
\end{iteMize}
\end{lem}
In particular, if $S \subseteq T$, then $\dim(S) \leq \dim(T)$.

A set $V \subseteq \mathbb R^n$ is called an \emph{(algebraic) variety}
if it can be defined as a conjunction of the form $p_1 = 0 \wedge \dots \wedge p_m = 0$ where $p_1,\dots,p_m$ are polynomials in the variables $x_1,\dots,x_n$ with coefficients from $\mathbb R$.
We allow terms in
polynomials to have degree zero.
 
\begin{lem}\label{lem:algebra}
Let $V \subseteq {\mathbb R}^n$ be a variety 
and let $L \subseteq {\mathbb R}^n$ be a line. If infinitely many points of $L$ are in $V$, then $L \subseteq V$. 
\end{lem}
\begin{proof}
Let $V$ be defined by 
$p_1(x_1,\ldots,x_n) = 0 \wedge \dots \wedge p_m(x_1,\ldots,x_n) = 0$,
and let $l_1,\dots,l_n$ be univariate linear polynomials 
such that $L=\{(l_1(x),\dots,l_n(x)) \; | \; x \in {\mathbb R}\}$. 
For each $p_i$, the univariate polynomial $p_i(l_1(x), \dots, l_n(x))$ equals 0 infinitely often.
So it is always 0, and it follows that every point on $L$ satisfies 
$p_1=0 \wedge \dots \wedge p_m = 0$. 
\end{proof}

\begin{thm}[Tarski and Seidenberg; Proposition 5.2.2 in~\cite{BochnakCosteRoy}]\label{thm:SA-QE}
Every first-order formula 
over $({\mathbb R};*,+,\leq)$ with parameters from $\mathbb R$ is equivalent to a quantifier-free formula with parameters from $\mathbb R$.
\end{thm}

By an \emph{interval} we mean either an open, half-open, or closed
interval with more than one element.
An ordered structure $(D; \leq,\ldots)$ is
{\em o-minimal} (see~\cite{Marker}, Definition 3.1.18) if for any first-order definable $S \subseteq D$ 
with parameters from $D$
there are finitely many intervals $I_1,\ldots,I_m$ with
endpoints in $D \cup \{\pm \infty\}$ and
a finite set $D_0 \subseteq D$ such that
$S = D_0 \cup I_1 \cup \cdots \cup I_m$.
The following is an easy and well-known consequence of Theorem~\ref{thm:SA-QE}.

\begin{thm}[see e.g.~\cite{Marker}]\label{thm:o-min}
Let $R_1,\dots,R_n$ be semi-algebraic relations. Then $({\mathbb R};\leq, R_1,\dots,R_n)$ is o-minimal. 
\end{thm}

A set $S \subseteq \mathbb R^n$ is called \emph{semi-linear} if it has a first-order definition in $(\mathbb R; +,\leq)$ with parameters from $\mathbb Q$; we also call first-order formulas over $(\mathbb R; +,\leq)$
with parameters from $\mathbb Q$ \emph{semi-linear}.
It has been shown 
in~\cite{DumortierGyssensVardeurzenVanGucht,DumortierGyssensVardeurzenVanGuchtErr} that it is decidable whether a given first-order
formula over $(\mathbb R; *,+,\leq)$ with parameters from $\mathbb Q$ defines a semi-linear relation or not. 
A set $V \subseteq {\mathbb R}^n$
is called a \emph{linear set}
if it can be defined as a conjunction of the form $p_1 \geq 0 \wedge \dots \wedge p_m \geq 0$ where $p_1,\dots,p_m$ are linear polynomials in the variables 
$x_1,\dots,x_n$ with coefficients from $\mathbb Q$. 
It is not hard to see that every semi-linear relation $S$ can be viewed as a finite union
of linear sets. We also have quantifier elimination for
semi-linear relations.




\begin{thm}[Ferrante and Rackoff~\cite{FerranteRackoff}]\label{thm:LP-QE}
Every semi-linear relation has a quantifier-free definition over 
$({\mathbb R};+,-,\leq)$ with parameters from $\mathbb Q$.
\end{thm}

\subsection{Definability of Rational Expressions}
\label{sect:pp}

The following elementary lemma will be needed for the observation that the
feasibility problem for linear programs is polynomial-time equivalent
to $\Csp(\lingamma)$; it is also essential for the hardness proofs in Section~\ref{sect:hard} and
for proving the dichotomy result for metric temporal constraint reasoning.

\begin{lem} \label{lem:pp-rational}
Let $n_0,n_1,\ldots,n_l \in {\mathbb Q}$ be rational numbers.
Then the relation $\{(x_1,\ldots,x_l) \; | $
$n_1x_1 + \cdots + n_l x_l = n_0\}$ is pp-definable 
in $(\mathbb R; \{(x,y,z) \; | \; x+y=z\},\{1\})$. Furthermore, the pp-formula that
defines the relation can be computed in polynomial time. 
\end{lem}
\begin{proof}
We first note that we can assume that $n_0,n_1,\dots,n_l$ are  
integers. To see this, suppose that the rational coefficients $n_0,\ldots,n_l$ are represented as pairs 
of integers $(a_0,b_0),\ldots,(a_l,b_l)$ where
$a_i$ denotes the nominator and $b_i$ the denominator.
Let $c=\prod_{i=0}^l b_i$ and create a new sequence of coefficients
$n'_0,\ldots,n'_l = (a_0 \cdot c / b_0,1),\ldots,(a_l \cdot c / b_l,1)$. The resulting
equation is obviously equivalent. 
It is also clear that it only takes polynomial time
to compute such coefficients. 


Before the actual proof, we note that
$x=0$ is pp-definable by $x+x=x$, 
and we therefore freely use the terms $0$ and $1$ in
pp-definitions. Similarly, $x=-1$ is pp-definable by $x+1=0$.
The proof is by induction on $l$.
We first show how to express equations of the form $n_1 x_1 + n_2 x_2  
= x_3$. By setting $x_2$ to $-1$ and $x_3$ to $0$, this will solve  
the case $l=1$. 
For positive $n_1,n_2$,
the formula  $n_1 x_1 + n_2 x_2  
= x_3$ is equivalent to
\begin{align*}
\exists u_1,\dots,u_{n_1}, v_1,\dots,v_{n_2}. \quad & u_1=x_1 \wedge  
\bigwedge_{i=1}^{n_1-1}
x_1+u_i=u_{i+1} \\
\wedge \; & v_1=x_2 \wedge \bigwedge_{i=1}^{n_2-1} x_2 + v_i  = v_{i+1}  \\
\wedge \; & u_{n_1}+v_{n_2} = x_3 \; .
\end{align*}
However, this formula
is exponential in the representation size of $n_1$ and $n_2$,
and cannot be used in polynomial-time reductions.

Let ${\it bit}(n,i)$ denote the $i$-th lowest bit in the  
binary
representation of an integer $n$ and $1 \leq i \leq \lfloor \log n \rfloor + 1 
$. The following formula is equivalent to the previous one (we are still in the case that both $n_1$ and $n_2$ are positive)
and has polynomial length in the representation size
of $n_1$ and $n_2$. Write $m_1$ for 
$\lfloor \log n_1 \rfloor+1$ and $m_2$ for
$\lfloor \log n_2 \rfloor+1$.

\begin{align*}
\exists a_1,...,a_{m_1}, b_1,...,b_{m_2}, c_1,...,c_{m_1}, d_1,...,d_{m_2}. \quad
& a_1 = x_1 \wedge \bigwedge_{i=2}^{m_1}
a_{i-1}+a_{i-1}=a_{i}\\
 \wedge  \; & b_1=x_2 \wedge \bigwedge_{i=2}^{m_2}
b_{i-1} + b_{i-1}  = b_{i}  \\
\wedge \; & c_1= {\it bit}(n_1,1) a_1 \wedge \bigwedge_{i=2}^{m_1}  {\it bit}(n_1,i) a_{i} + c_{i-1} = c_i   \\
\wedge \; & d_1= {\it bit}(n_2,1) b_1 \wedge \bigwedge_{i=2}^{m_2}  {\it bit}(n_2,i) b_{i} + d_{i-1} = d_i  \\
\wedge \; & c_{m_1} + d_{m_2}
= x_3 
\end{align*}


If $l=2$, and $n_1=0$ or $n_2=0$, then the proof is similar.
If $n_1$ and $n_2$ have different signs, we replace the conjunct
$c_{m_1} + d_{m_2} = x_3$ 
in the  formula above appropriately 
by $c_{m_1} + x_3 =  
d_{m_2}$ or $d_{m_2} + x_3 =  
c_{m_1}$. 
If both $n_1$ and $n_2$ are negative,
then we use the pp-definition $\exists x_3'. (- n_1 x_1 - n_2 x_2 = 
x_3' \wedge x_3' + x_3 = 0)$.

Equalities of the form $n_1 x_1 + n_2 x_2 = n_0$ can be defined
by $\exists x_3. (n_1 x_1 + n_2 x_2 = x_3 \wedge x_3 = n_0)$.
Now suppose that $l>2$. By the inductive assumption,
there is a pp-definition $\phi_1(x_1,x_2,u)$ for
$n_1x_1 + n_2x_2 + u = n_0$ and a pp-definition $\phi_2(x_3,\dots,x_l,u)$ for $n_3x_3+ \cdots + n_l  
x_l = u$.
Then $\exists u. (\phi_1 \wedge \phi_2)$
is a pp-definition for $n_1 x_1 + \cdots + n_l x_l = n_0$.
It is clear that the pp-definition given above can be computed in time which
is polynomial in the number of bits needed to represent the input. 
\ignore{We next prove that the pp-definition can be computed in
time $O(p(m))$ time where $p$ is a fixed polynomial and $m$ denotes
the number of bits needed to represent the input.
Let $T(l,m)$ denote the maximum time needed to compute the definition
for the equations containing $l$ variables and where $m$ bits are needed
for representing the input.
We assume without loss of generality that $l \leq m$.
If $1 \leq l \leq 3$, then $T(l,m) \leq p'(m)$ for some
polynomial $p'$ since we can compute $bit(\cdot,\cdot)$ in polynomial time.
If $l>3$, then $T(l,m) \leq T(3,m)+T(l-2,m)+p''(m)$ for some
polynomial $p''$ by the inductive construction of the definition. 
Assume for simplicity that $l$ is even; the case when $l$ is odd can
be solved analogously. By unfolding the recursive
definition of $T$, we see that
\begin{align*} 
T(l,m) \leq \frac{l}{2}T(3,m)+ p'(m) + \sum_{j=1}^{l/2} p''(m)
\leq & \; l \cdot p'(m) + p'(m) + l \cdot p''(m) \\
 \leq & \; m \cdot p'(m) + p'(m) + m \cdot p''(m)
\end{align*}
which is a polynomial in $m$.
}
 \end{proof}

By extending the previous result to inequalities, we prove
that $\Csp(\lingamma)$ and linear program feasibility are
polynomial-time equivalent problems. The dichotomy for
metric temporal reasoning follows immediately
by combining this result and Theorem~\ref{LPresult}.

\begin{prop}\label{prop:lp-pp}
The linear program feasibility problem 
is polynomial-time equivalent to $\Csp(\lingamma)$.
\end{prop}
\begin{proof}
It is clear that an instance of $\Csp(\lingamma)$ can be seen as a linear program feasibility problem,
since the three different relations in the constraint language, 
$x+y=z$, $x=1$, $x \leq y$, are linear.

For the opposite direction, let $\Phi$ be an arbitrary instance of 
the linear program feasibility problem. 
Given a linear equality $L(x_1,\dots,x_k) \equiv c_1x_1 + \cdots + c_k x_k = c_0$, let $\phi_{L(x_1,\dots,x_k)}$ 
denote the pp-definition of $L(x_1,\dots,x_k)$ 
in $(\mathbb R; \{(x,y,z) \; | \; x+y=z\},\{1\})$
obtained in Lemma~\ref{lem:pp-rational}.
Construct an instance $\Psi$ of $\Csp(\lingamma)$ by replacing each occurrence
of a linear inequality constraint $c_1x_1+ \cdots c_lx_l \leq c_0$ 
in $\Phi$ by a 
$\phi_{c_1x_1+\dots+c_lx_l-y=0} \wedge y \leq c_0$; use
fresh variables for $y$ and for the existentially quantified variables
introduced by $\phi_L$.
The resulting formula $\Psi$ can be rewritten as a primitive positive sentence over $\Gamma$ without increasing its length and,
by Lemma~\ref{lem:pp-rational}, the length of $\Psi$ is
polynomial in the length of $\Phi$. Since $\Phi$ is satisfiable if
and only if $\Psi$ is satisfiable, this shows that the problems
are polynomial-time equivalent. 
\end{proof}

\section{Hardness}\label{sect:hard}
We consider relations that give rise to NP-hard CSPs
in this section. We first need some definitions:
a relation $S \subseteq {\mathbb R}^k$ is {\em convex} if
for all $p,q \in S$, $S$ contains all points on the line
segment between $p$ and $q$.
We say that a relation $S \subseteq {\mathbb R}^k$ {\em excludes an interval} if there are
$p,q \in S$ and real numbers $0 < \delta_1 < \delta_2 < 1$ such that
$p+(q-p)y \not\in S$ whenever $\delta_1 \leq y \leq \delta_2$.
Note that we can assume that
$\delta_1,\delta_2$ are rational numbers, since we can choose
any two distinct rational numbers $\gamma_1<\gamma_2$ between $\delta_1$ and $\delta_2$
instead of $\delta_1$ and $\delta_2$.

\begin{defi}
We say that $S \subseteq \R^n$ is \emph{essentially convex} if for all $p,q \in S$ there are only finitely many points on
the line segment between $p$ and $q$ that are not in $S$. 
\end{defi}

If $S$ is \emph{not} essentially convex, and if
 $p$ and $q$ are such that there are infinitely many points on the line
 segment between $p$ and $q$ that are not in $S$, 
 then $p$ and $q$ {\em witness} that $S$ 
 is not essentially convex.
The following is a direct consequence of Theorem~\ref{thm:o-min},
and we will use it in the following without further reference.

\begin{cor}\label{cor:excludes}
If $S$ is a semi-algebraic relation that is not essentially convex,
then $S$ excludes an interval. If $S$ is an essentially convex semi-algebraic relation, 
and $a,b$ are two distinct points from $S$, then the line segment between $a$ and $b$ contains
an interval $I$ with $I \subseteq S$.
\end{cor}

The next proposition will be used several times in the sequel; it
clarifies the relation
between finite unions of varieties and essentially convex relations.

\begin{prop}\label{prop:indep}
Let $W$ be a finite union of varieties $V_1,\dots,V_k \subseteq {\mathbb R}^n$, and let $C \subseteq W$ be essentially convex.
Then, there is an $i \leq k$ such that $C\subseteq V_i$.
\end{prop}
\begin{proof}
Let $J \subseteq \{1,\dots,k\}$ be minimal such that 
$C \subseteq \bigcup_{i \in J} V_i$. If $|J|=1$, then there is nothing to show.
So suppose for contradiction that  there are distinct $i,j \in J$.
Then there must be points $a,b \in C$ such that $a \in V_i$
and $a \notin V_l$ for all $l \in J \setminus \{i\}$,
and $b \in V_j$ and $b \notin V_l$ for all $l \in J \setminus \{j\}$.
By essential convexity of $C$ and Corollary~\ref{cor:excludes}, 
the line segment $L$ between $a$ and $b$ must contain an interval $I$ that lies in $C$. 
Since $J$ is finite, there must be $l \in J$ such that infinitely many points on $I$ are from $V_l$.
By Lemma~\ref{lem:algebra}, all points on the line through $a$ 
and $b$ are from $V_l$; this contradicts the choice of $a$ and $b$.
 \end{proof}

The rest of the section is divided into two parts. We first prove that
if $S \subseteq {\mathbb R}^k$ is a semi-algebraic relation that
is not essentially convex and this
is witnessed by two rational points $p$ and $q$, then
$\Csp((\lingamma, S))$ is NP-hard. 
In the second part, we prove that if
$S \subseteq {\mathbb R}^k$ is a \emph{semi-linear} relation that
is not essentially convex, then this is witnessed by rational points and,
consequently, $\Csp((\lingamma,S))$ is NP-hard. 

\subsection{Semialgebraic relations and rational witnesses}

We begin with the special case when $S$
is a unary relation.
The hardness proof is by a reduction from $\Csp((\{0,1\}; R_{1/3}))$ 
where $$R_{1/3}=\{(1,0,0),(0,1,0),(0,0,1)\} \; .$$
This NP-complete problem is also called 
{\sc Positive One-In-Three 3Sat} \cite[LO4]{GareyJohnson}, which is the variant of
{\sc One-In-Three 3Sat} where we have the extra requirement that in all input instances of the problem, no clause contains a negated literal.

\begin{lem} \label{zero-one-rel}
Let $S \subseteq {\mathbb R}$ be a unary relation.
If $S$ excludes an interval and this is witnessed by rational 
points $p$ and $q$, 
then $\Csp((\lingamma, S))$ is
NP-hard.
\end{lem}

\begin{proof}
We know that there are rational numbers $0 < \delta_1 < \delta_2 < 1$ such that
$p+(q-p)y \not\in S$ whenever $\delta_1 \leq y \leq \delta_2$.
Let

\[a=\sup \{\delta_2-\delta_1 \; | \; \mbox{$0 < \delta_1 < \delta_2 < 1$ and $[p+(q-p)\delta_1,
p+(q-p)\delta_2] \cap S = \emptyset$}\},\]

\noindent
i.e., the least upper bound 
 on the length (scaled to the interval [0,1]) of excluded 
intervals between
$p$ and $q$.
Choose rational numbers $\delta_1,\delta_2$ such that

\begin{iteMize}{$\bullet$}
\item
there exists $y \in [\delta_1-d,\delta_1]$ such that $p+(q-p)y \in S$; and

\item
there exists $y \in [\delta_2,\delta_2+d]$ such that $p+(q-p)y \in S$.

\item
$S \cap [p+(q-p)\delta_1,p+(q-p)\delta_2] = \emptyset$.
\end{iteMize}

\noindent
where $d=(\delta_2-\delta_1)/5$.
It is easy to see that such $\delta_1,\delta_2$ exist;
we simply need to find $\delta_1,\delta_2$ such that $S \cap [p+(q-p)\delta_1,p+(q-p)\delta_2] = \emptyset$
and $\delta_2-\delta_1$ is sufficiently close to $a$. Clearly, for any $\epsilon > 0$,
there exist suitable $\delta_1,\delta_2$ such that $a-(\delta_2-\delta_1) < \epsilon$.

Now, define $p'=p+(q-p)(\delta_1-d)$, $q'=p+(q-p)(\delta_2+d)$, and
\begin{align*}
U(y) \equiv \exists z.
(z = p' + (q'-p') y \; \wedge \; S(z) \; \wedge \; 0 \leq y \leq 1).
\end{align*}

\noindent
Observe that $U$ is pp-definable in $\lingamma \cup \{S\}$ 
by Lemma~\ref{lem:pp-rational} combined by the fact that $p'$ and $q'$ are rational numbers.
We claim that
$U$ contains at least one point in the interval
$[0,d']$, at least one point in the interval $[1-d',1]$,
and no points in the interval $[d',1-d']$ where $d'=1/7$.
Let us consider the interval $[0,d']$. The point (expressed in $p$ and $q$) corresponding
to $y=0$ is $p'$ (which equals $p+(q-p)(\delta_1-d)$) while the point
corresponding to $y=1/7$ is

\begin{align*}
p'+\frac{(q'-p')}{7} & =p+(q-p)(\delta_1-d)+\frac{p+(q-p)(\delta_2+d)-p-(q-p)(\delta_1-d)}{7}\\
& =p+(q-p)(\delta_1-d)+\frac{(q-p)((\delta_2+d)-(\delta_1-d))}{7}\\
& =p+(q-p)(\delta_1-d)+\frac{(q-p)(\delta_2-\delta_1+2d)}{7}\\
& =p+(q-p)(\delta_1-d)+\frac{(q-p)(5d+2d)}{7}\\
& =p+(q-p)(\delta_1-d)+(q-p)d\\
& =p+(q-p)\delta_1
\end{align*}

We know that the choice of $\delta_1$ and $\delta_2$ implies that 
there is at least one point in $S$ on the
line segment between
$p+(q-p)(\delta_1-d)$ and $p+(q-p)\delta_1$. The other two cases can
be proved similarly.

We 
show NP-hardness by a
polynomial-time reduction from $\Csp((\{0,1\}; R_{1/3}))$.
Let $\phi$ be an arbitrary instance of this problem and let $V$ denote the set
of variables appearing in $\phi$.
Construct a formula 
\[\psi \equiv \bigwedge_{v \in V} U(v) \wedge \bigwedge_{R_{1/3}(v_i,v_j,v_k) \in \phi}
v_i+v_j+v_k \geq \frac{6}{7} \; \wedge \bigwedge_{R_{1/3}(v_i,v_j,v_k) \in \phi} v_i+v_j+v_k \leq \frac{11}{7}.\]
Lemma~\ref{lem:pp-rational} implies that $\psi$ is pp-definable in $(\mathbb R; \{(x,y,z) \; | \; x+y=z\},\{1\},\leq,U)$
(and, consequently, pp-definable in $(\mathbb R; \{(x,y,z) \; | \; x+y=z\},\{1\},\leq,S)$)
and the formula can be constructed in polynomial time.
We now verify that the formula $\psi$
has a solution if and only if $\phi$ has a solution.

Assume that there exists a satisfying truth assignment $f \colon V \rightarrow \{0,1\}$ to the
formula $\phi$. We construct a solution $g$ for $\psi$ as follows: arbitrarily choose
a point $t_0 \in [0,d']$ such that $t_0 \in U$ and a point
$t_1 \in [1-d',1]$ such that $t_1 \in U$. Let $g(v)=t_0$ if $f(v)=0$ and
$g(v)=t_1$, otherwise. Clearly, this assignment satisfies every literal of
the type $U(v)$.
Each literal $v_i+v_j+v_k \geq 6/7$ is satisfied, too, since $g(v_i)+g(v_j)+g(v_k)=2 \cdot t_0+t_1 \geq 2 \cdot 0 + (1-d') = 1-d'=6/7$. Similarly, each literal $v_i+v_j+v_k \leq 11/7$
is also satisfied:  $g(v_i)+g(v_j)+g(v_k)=2 \cdot t_0+t_1 \leq 2 \cdot d'+1 = 9/7$.

Assume now instead that there exists a satisfying assignment $g \colon V \rightarrow {\mathbb R}$
for the formula $\psi$. Each variable obtains a value that is in either the
interval $[0,d']$ or in the interval $[1-d',1]$.
If a variable is assigned a value in $[0,d']$, then we consider
this variable `false', i.e., having the truth value 0; analogously,
variables assigned values in $[1-d',1]$ are considered `true'.

We continue by looking at an arbitrary literal
$R_{1/3}(v_i,v_j,v_k)$ in $\phi$ and its corresponding
inequalities (1) $v_i+v_j+v_k \geq 6/7$ and (2) $v_i+v_j+v_k \leq 11/7$.
If all three variables are assigned values within $[0,d']$, then their sum
is at most $3d'=3/7$ which violates inequality (1).
If two of the variables appear within $[1-d',1]$, then their sum
is a least $0+2(1-d')=12/7$ which violates inequality (2); naturally, this inequality
is violated if all three variables appear within $[1-d',1]$, too.
If exactly one variable appears within $[1-d',1]$, then
the sum of the variables is at least $1-d'=6/7$ and at most $1+2d'=9/7$.
We see that both inequality (1) and (2) are satisfied.
Hence, we can define a satisfying assignment $f \colon V \rightarrow \{0,1\}$ for $\phi$:

\[ f(v) = \left\{ \begin{array}{ll}
  0 & \mbox{\; \;if $0 \leq g(v) \leq d'$} \\
  1 & \mbox{\; \; otherwise}
\end{array}
\right. \]
\noindent
This concludes the proof.
\end{proof}

It is now straightforward to lift Lemma~\ref{zero-one-rel} to
relations with arbitrary arities.

\begin{lem} \label{lem:semialghard}
Let $S \subseteq {\mathbb R}^k$ be a semi-algebraic relation that
is not essentially convex, and this
is witnessed by two rational points $p=(p_1,\ldots,p_k)$ and $q=(q_1,\ldots,q_k)$. Let
$\Gamma$ be the structure $(\lingamma, S)$. Then, $\Csp(\Gamma)$ is NP-hard. 
\end{lem}
\begin{proof}
Define 
\begin{align*}
U(y) \quad \equiv \quad \exists \bar z. \; 
\bigwedge_{i=1}^k z_i = p_i + (q_i-p_i) y \; \wedge \; S({\bar z}) \; \wedge \; 0 \leq y \leq 1
\end{align*}
where ${\bar z}=(z_1,\ldots,z_k)$.
By Corollary~\ref{cor:excludes}, $U$
excludes an interval and $\Csp(\Gamma)$ is NP-hard by 
Lemma~\ref{zero-one-rel} since $U$ is pp-definable in $\Gamma$.
 \end{proof}

\begin{rem}
If $S$ is not essentially convex and this is witnessed by non-rational points
only, then the problem $\Csp(\Gamma)$ for $\Gamma = ({\mathbb R}; \{(x,y,z) \; | \; x+y=z\},\{1\},\leq,S)$ might
still be solvable in polynomial time. Consider for instance 
the binary relation $$S = \{(x,y) \; \big | \; (|x+y| \leq 1) \wedge (y = \sqrt{2} x \rightarrow |x+y| = 1) \} \; .$$
Clearly, $S$ is not essentially convex; however, the only witnesses are
$(\sqrt{2} - 1,2 - \sqrt{2})$ and $(-\sqrt{2}+1, -2 + \sqrt{2})$ (see Figure~\ref{fig:example}).

\begin{figure}[h]
\begin{center}
\includegraphics[scale=0.6]{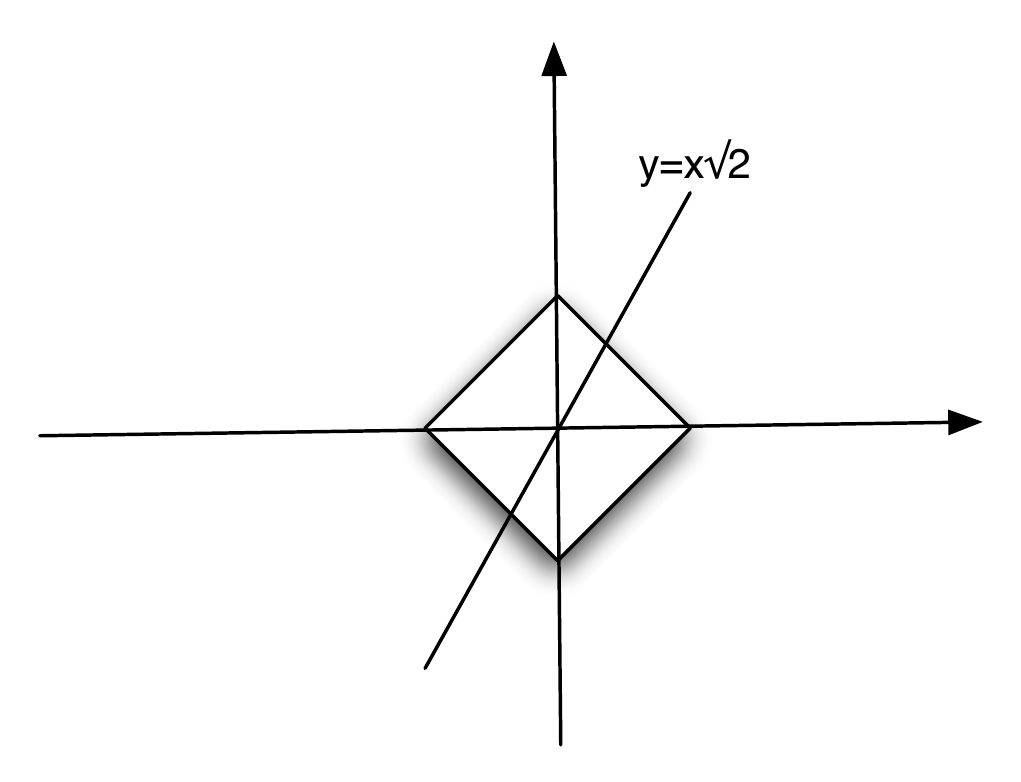}
\end{center}
\caption{Illustration of relation $S$}
\label{fig:example}
\end{figure}

We show that $\Csp(\Gamma)$ can be solved in polynomial time. To see this,
define 
\[S'=\{(x,y) \; | \; (|x+y| \leq 1) \wedge (x \neq 0 \vee y \neq 0)\}\] 
and
$\Delta = ({\mathbb R}; \{(x,y,z) \; | \; x+y=z\},\{1\},\leq,S'\}$.
We first prove that
a primitive positive sentence is true in $\Gamma$ if and only if it is true in $\Delta$.
Clearly, if a primitive positive sentence is true in $\Gamma$, then it is also true in $\Delta$, since the relations in $\Delta$ are supersets of the corresponding relations in $\Gamma$. 
Conversely, suppose that $\Phi$ is primitive positive and
true in $\Delta$. Let $\alpha$ be an assignment of the variables
of $\Phi$ that satisfies all conjuncts in $\Phi$. 
Since $\Delta$ is semi-linear, we can assume that $\alpha$ is rational
(see Lemma~\ref{alwaysrational}).
The only relations that are different in $\Gamma$ and in $\Delta$ are the relations $S$ and $S'$. Since $S' \setminus S$ contains irrational points only, the assignment $\alpha$ shows that $\Phi$ is also
true in $\Delta$.
Finally, $\Csp(\Delta)$ can be solved in polynomial time by the results in Section~\ref{sect:semilinear} (note that
the constraint $|x+y| \leq 1$ is equivalent to a conjunction of four linear inequalities). 
\end{rem}

\subsection{Semilinear relations}


In the previous section, we showed that there exists a semi-algebraic relation $S$ that is not essentially convex, but  $\Csp(({\mathbb R}; \{(x,y,z) \; | \; x+y = z\},\{1\},\leq,S))$ is polynomial-time solvable. 
If we restrict ourselves to semi-linear relations $S$, then this phenomenon cannot occur: 
indeed, in this section we prove that if a semi-linear relation $S$ is not essentially convex, then this is witnessed by rational points (Lemma~\ref{semilinhard}), 
and $\Csp(({\mathbb R}; \{(x,y,z) \; | \; x+y = z\},\{1\},\leq,S))$ is NP-hard.

\begin{lem} \label{alwaysrational}
Every non-empty semi-linear relation $S$ contains at least one
rational point.
\end{lem}

\begin{proof}
Assume first that $S$ is a non-empty unary relation such that $S \cap {\mathbb Q}=\emptyset$.
If $S$ contains infinitely many points, then it also contains an interval due to $o$-minimality of $S$;
this contradicts that $S \cap {\mathbb Q} = \emptyset$. So we assume that
$S$ contains a finite number of points. Consider
the unary relation $S'= \{\min(S)\}$ and note that it can be pp-defined in 
$(\lingamma,S)$ by
$S'(x) \equiv S(x) \wedge x \leq p$ where $p$ denotes a suitably chosen rational number.
By Theorem~\ref{thm:LP-QE}, $S'$ has a quantifier-free definition $\phi$ over $({\mathbb R};+,-,\leq)$ with parameters from $\mathbb Q$,
and we can without loss of generality assume that $\phi$ is in disjunctive normal
form, and contains a single disjunct since $|S'|=1$. 
Assume without loss of generality that every conjunct of this disjunct of $\phi$ 
is of one of the following forms: $x \geq c$, $x \leq c$, or $x \neq c$ (where $c$ denotes some rational number).
Let $a=\max\{c \; | \; (x \geq c) \in \phi\}$ and
$b=\min\{c \; | \; (x \leq c) \in \phi\}$.
If $a=b$ then $S'=\{a\}$ and we have a contradiction since
$a$ is a rational number.
If $a<b$, then $S'$ contains an infinite number of points
(regardless of the number of disequality constraints in $\phi$)
and we have a contradiction once again. 

Assume now that $ar(S)=d>1$.
Arbitrarily choose
a point $s=(s_1,\ldots,s_d) \in S$ with a maximum number of
rational components. 
Assume without loss of generality that $s_1,\ldots,s_{k'}$, $k' < k$ are rational
points and consider the unary relation
\[U(x_k) \equiv \exists x_1,\ldots,x_{k-1}. (S(x_1,\ldots,x_k) \; \wedge \; x_1=s_1 \wedge \cdots \wedge x_{k'}=s_{k'}).\]
We get a contradiction since $U \cap {\mathbb Q} = \emptyset$, $U$ is non-empty, and $U$ is semi-linear. 
\end{proof}

\begin{cor} \label{rationalball}
Let $S \subseteq {\mathbb R}^k$ be a semi-linear relation and let $s \in S$
be arbitrary.
Then, every open $k$-dimensional ball $B$ around $s$ of radius $\epsilon > 0$ contains a rational point in $S$.
\end{cor}
\begin{proof}
If there is an $\epsilon$ such that $B$ does not contain any rational point in $S$, then there is
a linear set $P$ within $B$ such that $S \cap P$
only contains irrational points. This contradicts Lemma~\ref{alwaysrational}.
\end{proof}

A {\em hyperplane} is a set $V=\{x \in {\mathbb R}^k \; | \; p(x)=0\}$ 
where $p$ is a linear term such that $\emptyset \subset V \subset {\mathbb R}^k$ (this makes sure that the degree of $p$ is one). 
We do not require that the coefficients in $p$ are rational; it is important to 
note that
this differs from the definition of a linear set. 
If all coefficients appearing in $p$ are rational, then we say that the hyperplane is \emph{rational}.

\begin{lem} \label{semilinhard}
If $T$ is a semi-linear relation that is not essentially convex,
then this is witnessed by rational points, and 
$\Csp((\Gamma_{lin},T))$ is NP-hard.
\end{lem}
\begin{proof}
If there are rational witnesses of the fact that $T$ is not
essentially convex, 
then NP-hardness follows from
Lemma~\ref{lem:semialghard} and we are done.

Assume now that there exists a relation $T$ that is
not essentially convex but $T$ lacks rational witnesses.
Arbitrarily choose such a $T$ with minimal arity $k$.
We first consider the case when $k=1$.
Arbitrarily choose witnesses $p,q \in T$. 
By $o$-minimality, 
there are finitely many intervals $I_1,\ldots,I_m$ with
endpoints in ${\mathbb R} \cup \{\pm \infty\}$ and
a finite set $D_0 \subseteq {\mathbb R}$ such that
$T = D_0 \cup \bigcup_{i=1}^m I_i$.
Now, apply the following process to $D_0$ and $I_1,\ldots,I_m$.

\begin{iteMize}{$\bullet$}

\item
if there is a point $d \in D_0$ and an interval $I_j$, $1 \leq j \leq m$,
such that $d$ is in $\overline{I_j}$, then
remove $d$ from $D_0$ and replace $I_j$ with $I_j \cup \{d\}$;

\item
repeat until $D_0$ is not changed.
\end{iteMize}

After these modifications, the sets $I_1,\ldots,I_m$ are
still (open, half-open, or closed) intervals, and for every point $d \in D_0$,
there exists an $\epsilon_d > 0$ such that 
$[d-\epsilon_d,d+\epsilon_d] \cap T = \{d\}$.

Assume without loss of generality that $p \not\in {\mathbb Q}$.
If $p \in D_0$, then choose rational numbers $p^-,p^+$ such that
$p-\epsilon_p < p^- < p < p^+ < p+\epsilon_p$; this is always possible
since the rationals are a dense subset of the reals. Consider
the semi-linear relation

\[T'(x) \equiv T(x) \wedge p^- \leq x \leq p^+\]
\noindent
and note $T'=\{p\}$. However, $p$ is not a rational number which
contradicts Lemma~\ref{alwaysrational}. We may thus assume that
$p \not\in D_0$ and that $p$ is a member of
an interval $I \in \{I_1,\ldots,I_m\}$. Arbitrarily choose
one rational point $p' \in I$;
once again, this is possible since the rationals are a dense subset of the reals.
Note that $p',q$ witness that $T$ is not essentially convex.
If $q \in {\mathbb Q}$, then we are done so we assume that $q \not\in {\mathbb Q}$.
We see that $q \not\in D_0$ by reasoning as above.
Consequently, $q$ is a member of
an interval $J \in \{I_1,\ldots,I_m\}$ and $I \neq J$. Finally choose
a rational point $q' \in J$ and note that $p',q'$ are rational points
witnessing that $T$ is not essentially convex.

\medskip

Assume instead that $k > 1$.
Let ${\mathcal S}_k$ denote the set of relations 
$S$ that satisfy 1, 2, and 3:

\begin{enumerate}
\item
$S$ is a semi-linear relation of arity $k$,

\item
$S$ is not essentially convex, and

\item
for every
pair of witnesses that $S$ is not essentially convex,
at least one 
is irrational.
\end{enumerate}

We now conclude the proof by considering two different cases.

\medskip

\noindent
{\bf Case 1:} There exists an $S \in {\mathcal S}_k$ and 
a finite set of rational hyperplanes 
$H_1,\dots,H_h$  such that $S \subseteq \bigcup_{j=1}^h H_j$.
Choose the hyperplanes such that $h$ is minimal.
Let $v,w \in S$
be arbitrarily chosen witnesses for the fact that
$S$ excludes an interval, and let $I$ denote this interval.

Suppose first that $h=1$, i.e., that there is a single 
hyperplane $H$ such that $S \subseteq H$.
Obviously,
$x=(x_1,\ldots,x_k) \in H \Leftrightarrow c_1x_1 + \cdots + c_kx_k = c_0$ for
some rational constants $c_0,\ldots,c_k$.
We assume without loss of generality that at least
one $c_i$, say $c_k$, is non-zero.
Define the relation $S'$ by

\[S'(x_1,\ldots,x_{k-1}) \equiv \exists y. (S(x_1,\ldots,x_{k-1},y) \wedge y=\frac{c_0-c_1x_1 - \ldots - c_{k-1}x_{k-1}}{c_k}) \; .\]

\noindent
Let $v'=(v_1,\ldots,v_{k-1})$ and $w'=(w_1,\ldots,w_{k-1})$, and
note that $v',w'$ are witnesses of an excluded interval in $S'$.
If $S'$ lacks rational witnesses of essential non-convexity,
then the fact that $S'$ has arity $k-1$ 
contradicts the choice of $T$.
Hence, $S'$ has two rational witnesses $t=(t_1,\ldots,t_{k-1})$ and $u=(u_1,\ldots,u_{k-1})$. This implies that
\[t'=\left(t_1,\ldots,t_{k-1},\frac{c_0-c_1t_1 - \ldots - c_{k-1}t_{k-1}}{c_k} \right)\]
and 
\[u'=\left(u_1,\ldots,u_{k-1},\frac{c_0-c_1u_1 - \ldots - c_{k-1}u_{k-1}}{c_k} \right)\]
are rational witnesses for $S$, which leads to a contradiction.

Next, suppose that $h \geq 2$.
Let $H_1'=S \cap (H_1 \setminus \bigcup_{j=2}^h H_j)$ and
$H_2'=S \cap (H_2 \setminus \bigcup_{j \in \{1,3,\ldots,h\}} H_j)$.
By the minimal choice of $h$, $H'_1$ and $H'_2$
are non-empty.
Furthermore, they are semi-linear so we can choose rational points
$p_i \in H'_i$, $1 \leq i \leq 2$, by Lemma~\ref{alwaysrational}.
We now claim that at most a finite number of points on
the line segment between $p_1$ and $p_2$ lie in $S$.
Suppose to the contrary that infinitely many points lie on the line segment. 
Then, there must be one $H_i$, $i \geq 1$, such that infinitely many points from
$H_i$ lie on the line segment. Hence, $H_i$ (since it is a variety)
must contain the entire line by Lemma~\ref{lem:algebra}. This leads to
 a contradiction since $p_1$ and $p_2$ are chosen so that
$|\{p_1,p_2\} \cap H_j| \leq 1$, $1 \leq j \leq h$.
Thus, we have found rational witnesses for essential non-convexity of $S$
and obtained a contradiction since $S \in {\mathcal S}_k$.

\medskip

\noindent
{\bf Case 2:} 
There is no $S \in {\mathcal S}_k$ such that there exists 
a finite set of rational hyperplanes
$H_1,\dots,H_h$ and $S \subseteq \bigcup_{j=1}^h H_j$.
Arbitrarily choose $S \in {\mathcal S}_k$, let $v,w \in S$
be arbitrarily chosen witnesses for the fact that
$S$ excludes an interval, and let $I$ denote such an interval.

If there exists a rational hyperplane $H$ 
such that $\{v,w\} \subseteq S \cap H$, then the semi-linear relation
\[S'(x_1,\ldots,x_k) \equiv S(x_1,\ldots,x_k) \wedge H(x_1,\ldots,x_k)\]
excludes an interval and this is witnessed by $v$ and $w$.
Obviously, $S' \in {\mathcal S}_k$ and $S' \subseteq H$.
This contradicts the assumptions for this case so
we assume that $\{v,w\}$ (and consequently $I$) do not lie on any
rational hyperplane.

Next, we prove a couple of facts.

\medskip

\noindent
{\em Fact 1:} $I \subseteq \bar{S} \setminus S$. We show that
there is no point $e \in I$ and an $\epsilon > 0$ such that
the open $k$-dimensional ball $B$ around $e$ with radius $\epsilon$ satisfies
$B \cap S = \emptyset$.
Assume to the contrary that there is a point $e \in I$
satisfying this condition.
By Corollary~\ref{rationalball}, there exist rational points in $S$ arbitrary close
to $v$ and $w$. Thus, one can find rational points $v',w' \in S$
such that the line segment $L$ between $v'$ and $w'$ passes through $B$
and $L'=L \cap B$ has non-zero length.
In other words, $v'$ and $w'$ are rational witnesses of an excluded interval and we have
obtained a contradiction.

\medskip

\noindent
{\em Fact 2:} There exists a finite set $\{H_1,\ldots,H_h\}$ of rational
hyperplanes such that
$\bar{S} \setminus S \subseteq \bigcup_{i=1}^h H_i$.
Let $\phi$ be a first-order definition of $S$ and let
$\psi = D_1 \vee \cdots \vee D_n$ 
be a quantifier-free definition of $S$ in disjunctive
normal form; such a $\psi$ exists due to Theorem~\ref{thm:LP-QE}. 
Note that every parameter appearing in $\psi$
is rational: initially, every parameter in $\phi$ is rational, the quantifier
elimination does not add any irrational parameters, and the conversion
to disjunctive normal form does not introduce any new parameters.
Let $l_1,\ldots,l_m$ denote the literals appearing
in $\phi$. 
For each 
$l_i \equiv p(x_1,\ldots,x_k) \: r \: 0$ (where $r \in \{\leq,<,=,\neq,>,\geq\}$), create
a hyperplane $H_i=\{(x_1,\ldots,x_k) \subseteq {\mathbb R}^k \; | \; p(x_1,\ldots,x_k)=0\}$. In other words, we let the boundary of the 
subspace defined by $l_i$ define
a hyperplane $H_i$.
It is now easy to see that $\bar{S} \setminus S \subseteq \partial S \subseteq \bigcup_{i=1}^m H_i$.
Furthermore, every hyperplane $H_1,\ldots,H_m$ is rational.
\medskip

\medskip

\noindent
We are now ready to prove the second case of the proof. By Fact 1, $I \subseteq \bar{S} \setminus S$. The set $\bar{S} \setminus S$ is a subset
of $\bigcup_{i=1}^h H_i$ where $H_1,\ldots,H_h$ are rational hyperplanes by Fact 2.
Hence,
$I$ is a subset of some $H_i$
by Proposition~\ref{prop:indep}, a contradiction. 
\end{proof}

\section{Essentially Convex Relations}
\label{sect:main}
Before we present a logical characterization of essentially convex semi-algebraic relations, we give examples that show that
two more naive syntactic restrictions of first-order formulas
are not powerful enough for defining all essentially convex semi-algebraic relations. Both of those restrictions are motivated by classes of
essential convex semi-linear relations that have appeared in the literature, cf.~\cite{JonssonBaeckstroem}.
When $S$ is a subset of $\mathbb R^n$, we write $\neg S$ for the
complement of $S$, i.e., for $\mathbb R^n \setminus S$.

We start with an example that shows that not every essentially convex semi-algebraic relation can be defined by 
conjunctions of first-order formulas of the form 
$$ p_1 \neq 0 \vee \dots \vee p_k \neq 0 \vee \phi $$
where $p_1,\dots,p_k$ are polynomials with coefficients from $\mathbb R$, and where $\phi$ defines a convex set.
It is easy to see that every relation that can be defined by such a 
conjunction is essentially convex.

See Figure~\ref{fig:examples}, left side. The figure shows a
1-dimensional variety $C \subseteq {\mathbb R}^2$, given as $\{ (p(t),q(t)) \; | \; t \in \mathbb R\}$ for polynomials $p$ and $q$.
The figure also shows 
two marked segments $S_1, S_2$ on the curve $C$. 
The marked segments are chosen such that
one end point is contained in interior of 
the convex hull of the other three end points of the segments. 

Let $S$ be the set $\neg C \cup S_1 \cup S_2$. Clearly, 
$S$ is essentially convex. Now, suppose for 
contradiction that $S$ has a definition $\psi$ as described above.
Let $H$ be the convex hull of $S_1 \cup S_2$.
The crucial observation is that the set $G := 
(H \cap C) \setminus (S_1 \cup S_2)$ is infinite.
Since no point from $G$ is in $S$, there must
be a conjunct  $p_1 \neq 0 \vee \dots \vee p_k \neq 0 \vee \phi$
in $\psi$ that excludes infinitely many points from $G$.
In particular, the variety $V$ defined by $p_1=\dots=p_k=0$ contains infinitely many points from $C$.
As in the proof of Lemma~\ref{lem:algebra}, one can see that $V$ must contain $C$.
Hence, all points in $S_1 \cup S_2$ must satisfy $\phi$;
but in this case, all points in $G$ satisfy
$p_1 \neq 0 \vee \dots \vee p_k \neq 0 \vee \phi$,
a contradiction.

\begin{figure}[h]
\begin{center}
\begin{tabular}{lp{4cm}r}
\includegraphics[scale=0.6]{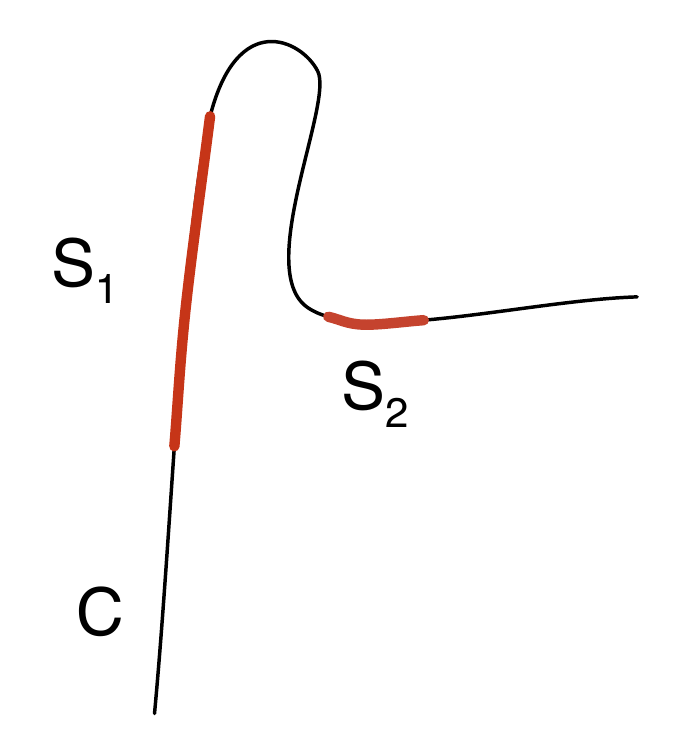}
& &
\includegraphics[scale=0.6]{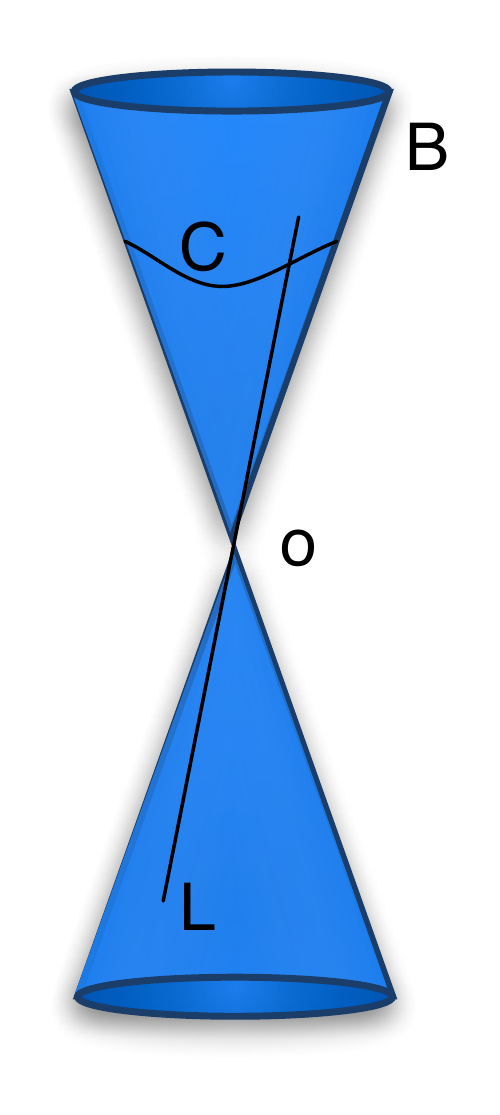}
\end{tabular}
\end{center}
\caption{Examples of essentially convex relations.}
\label{fig:examples}
\end{figure}

Our first example might motivate the following notion of definability:
we consider conjunctions of formulas of the form $p_1 \neq 0 \vee \dots \vee p_k \neq 0 \vee \phi$ such that for every conjunction of linear equalities $\psi$ that implies $p_i = 0$ for all $i \leq k$, the set defined
by $\phi \wedge \psi$ is convex. 
The set described above can indeed be defined by such a 
formula. Similarly as before, it is also easy to see that all relations that can be defined
in such a way are essentially convex. However, we again
have an example of a semi-algebraic essentially convex relation
that cannot be defined by such a conjunction.

See Figure~\ref{fig:examples}, right side. The figure shows
the boundary $B$ of a doubly infinite cone with apex $o$. 
On the boundary, there is a straight line segment $L$ through $o$,
and a circle $C$ that cuts $L$. Let $S$ be the set
$\neg B \cup C \cup (L \setminus \{o\})$. It can be verified that
$S$ is essentially convex. However, we claim that
there is no conjunction as described above that defines $S$.
The reason is that when 
$p_1 \neq 0 \vee \dots \vee p_k \neq 0 \vee \phi$ is such that
$p_1 = 0 \wedge \dots \wedge p_k = 0$ describes $B$,
and if $\phi$ contains $C \cup (L \setminus \{o\})$,
then it must also contain $o$ in order that $p_1 \neq 0 \vee \dots \vee p_k \neq 0 \vee \phi$ meets the required condition.
However, the set $\neg B \cup C \cup L$ is not essentially
convex since $o$ and points from the cycle exclude an interval.

The correct definition of formulas that correspond to essentially convex
sets has to take these examples into account. We call these formulas 
\emph{convex Horn formulas}. 
Basically, a convex Horn formula is a conjunction of
implications $p_1= \cdots =p_k = 0 \rightarrow \phi$ such that the premise defines a
variety $V$, and the formula $\phi$ is again convex Horn when
restricted to any convex subset of $V$. 
Formally, we have the following definition.

\begin{defi}
The set of convex Horn formulas  
is the smallest set of first-order formulas such that 
\begin{iteMize}{$\bullet$}
\item all formulas defining 
convex closed semi-algebraic relations over 
$(\mathbb R; +,*,\leq)$ with parameters from $\mathbb R$ are convex Horn;
\item Suppose that $p_1,\dots,p_k$ are polynomials, 
$\phi$ is a first-order formula
that defines a set $U \subseteq {\mathbb R}^n$,
and for every semi-algebraic convex set $C$ contained in the set defined by $p_1=\dots=p_k=0$,
the set $C \cap U$ can be defined by a convex Horn formula, and 
has strictly smaller dimension than the set defined by
$\psi \equiv (p_1 \neq 0 \vee \dots \vee p_k \neq 0 \vee \phi)$.
Then $\psi$ is also convex Horn.
\item Finite conjunctions of convex Horn formulas are convex Horn.
\end{iteMize}
\end{defi}

For example, the formula $(x^2 - y \neq 0) \wedge (y \geq 1)$,
describing
the half-plane above $y = 1$ with exception of the standard parabola,  is convex Horn. Every convex set $C$ contained in the set defined by $x^2-y=0$ consists of at most one point, and hence is $0$-dimensional and can be defined
by a convex Horn formula.


We can prove properties about the set of all convex Horn formulas
by induction over the \emph{level} of a convex Horn formula,
which is defined as follows. The level of a formula that defines
a convex closed semi-algebraic relation is $0$. 
Now, suppose we have already defined convex Horn formulas
of level smaller than $i$, and let $\psi$ be a convex Horn
formula that does not have level smaller than $i$.
Then $\psi$ has level $i$ if it is the finite conjunction
of formulas 
$\psi' \equiv (p_1 \neq 0 \vee \dots \vee p_k \neq 0 \vee \phi)$ 
such that 
for every semi-algebraic convex set $C$ contained in the set 
defined by $p_1=\dots=p_k=0$, the intersection of $C$
with the set defined by $\phi$ is convex Horn,
has level at most $i-1$, and strictly smaller dimension than the set defined by $\psi'$.
Since the intersection of sets of dimension $n$
has at most dimension $n$, it follows directly
from the definition of convex Horn formulas that
the level of a convex Horn formula $\phi$ is bounded
by the dimension of the set defined by $\phi$.

Looking back at the formula $(x^2 - y \neq 0) \vee (y \geq 1)$, we claim
that it is
convex Horn of level one: every convex set contained in the
parabola consists of at most one point, and can hence be described by a
convex Horn formula of level zero. For another example, 
consider $(x^2 - y \neq 0)
\vee (z \geq 0)$, i.e., the same parabola in three dimensions on one side
side of the $x$-$y$-plane. Each convex subset of $x^2 - y = 0$ is a point, a
straight line, or a line segment in the $z$ direction and can again be
described by a level zero convex Horn formula.

We are now ready to logically define essentially convex semi-algebraic
sets via convex Horn formulas. This is done in two steps;
we first prove (in Proposition~\ref{prop:linear-horn-convex})
that every set defined by a semi-algebraic convex Horn formula
is essentially convex. 
The rest of the section is devoted to proving the 
other direction---the final result can be found
in Theorem~\ref{thm:main}.

\begin{prop}\label{prop:linear-horn-convex}
Any set $S$ defined by a convex Horn formula $\psi$ over $(\mathbb R; *,+,\leq)$ is essentially convex.
\end{prop}
\begin{proof}
Our proof is by induction over the level of $\psi$.
Let $m$ denote the number of free variables in $\psi$.
If the level of $\psi$ is $0$, 
then $S=\{x \in {\mathbb R}^m \; | \; \psi(x)\}$ is a closed convex set
and, in particular, essentially convex. 

Assume that all relations defined by convex Horn formulas of level $< i$
are essentially convex. 
Arbitrarily choose a convex Horn formula 
$\psi \equiv p_1 \neq 0 \vee \dots \vee p_k \neq 0 \vee \phi$ 
with level $i$. Define $S = \{x \in {\mathbb R}^m \; | \; \psi(x)\}$,
$V=\{x \in {\mathbb R}^m \; | \; p_1(x)= \cdots =p_k(x)=0\}$, and
$U=\{x \in {\mathbb R}^m \; | \; \phi(x)\}$.
Since $\psi$ is level-$i$ convex Horn, we know that for every semi-algebraic
convex set $C$ such that $C \subseteq V$, the set $C \cap U$ 
can be defined by a convex Horn formula of level smaller than $i$ and
$\dim(C \cap U) < \dim(S)$.

Suppose for contradiction
that there are $a,b \in S$ 
and an infinite set $I$ of points on the line segment $L$ between $a$ and $b$ that is not contained in $S$.
In particular, $I \subseteq V$ since 
$S= \neg V \cup U$. By Lemma~\ref{lem:algebra}, 
all points on the line through $a$ and $b$ are in $V$.
However, $a$ and $b$ are in $S$ and therefore in $U$ so
$L \cap U$ is not essentially convex. 
We now note that $L$ is a semi-algebraic convex set that is a subset
of $V$ so $L \cap U$ can be defined by a convex Horn
formula of level smaller than $i$.
Consequently, $L \cap U$ is essentially convex by the inductive assumption
which leads to a contradiction. 

Finally, suppose that $\psi$ is a finite conjunction of convex
Horn formulas of level at most $i$. 
Since the intersection of finitely many 
essentially convex relations is essentially convex,
we are done also in this case.
\end{proof}

Next, we need some preparations for the proof of the converse implication (Theorem~\ref{thm:main}): 
we show that semi-algebraic relations
can be defined by a special type of formula (Lemma~\ref{lem:standard-def}),
and that the closure $\bar{S}$ of an essentially convex relation $S$ is 
convex (Lemma~\ref{lem:closure-convex}).

\begin{defi} \label{def:convexhorn}
Let $S$ be a semi-algebraic relation. 
We say that a first-order formula 
$\phi$ is a \emph{standard defi} of $S$ 
if 
\begin{iteMize}{$\bullet$}
\item $\phi(x_1,\dots,x_k)$ defines $S \subseteq {\mathbb R}^k$ 
over $({\mathbb R};*,+,-,\leq)$ with parameters from $\mathbb R$;
\item $\phi$ is in quantifier-free conjunctive normal form; 
\item if we remove any literal from $\phi$, then the resulting
formula is not equivalent to $\phi$; and
\item all literals are of the form $t \leq 0$ or $t \neq 0$.
\end{iteMize}
\end{defi}

\begin{lem}\label{lem:standard-def}
Every semi-algebraic relation $S$ has a standard definition. If $S$ is even semi-linear, then it has a standard definition $\phi$ that does not involve the function symbol for multiplication and irrational parameters.
\end{lem}
\begin{proof}
By Theorem~\ref{thm:SA-QE}, we know that $S$ has a quantifier-free definition over $(\mathbb R; *,+,\leq)$ with parameters from $\mathbb R$, and it is clear
that such a definition can be rewritten in conjunctive normal form $\phi$. Replace a clause $\alpha$ in $\phi$ with a literal of the form $a < b$ by two clauses $\alpha_1$ and $\alpha_2$ obtained from $\alpha$ by replacing $a < b$ by $a \leq b$ and by $a \neq b$, respectively.
In the same way we can eliminate occurrences of $a = b$ from $\phi$
using $\leq$. Literals of the form $a \leq b$ ($a \neq b$) can then be replaced by $a-b \leq 0$ (and $a-b \neq 0$, respectively).
Finally, we remove literals from $\phi$ as long as the resulting formula is equivalent to the original formula. 

If $S$ is semi-linear, then by Theorem~\ref{thm:LP-QE} we know that
$S$ has a quantifier-free definition over $(\mathbb R;+,-,\leq)$ with parameters from $\mathbb Q$,
and it is then clear that the formula constructed from $\phi$ as above
will be a standard definition of $S$ without the function symbol for multiplication and irrational parameters.
\end{proof}

\begin{lem}\label{lem:closure-convex}
The closure $\bar S$ of an essentially convex relation $S$ is
convex.
\end{lem}
\begin{proof}
Let $a,b \in \bar S$. We will show that all points $c$ on the line
segment between $a$ and $b$ are in $\bar S$. We have to
show that for every $\epsilon > 0$ there is a point $c'$ in $S$
such that the distance between $c$ and $c'$ is smaller than $\epsilon$. Since $a \in \bar S$ and $b \in \bar S$, there are points
$a' \in S$ and $b' \in S$ that are closer than $\epsilon/2$ to $a$ and $b$, respectively. Let $L$ be the line from $a'$ to $b'$. It is clear
that there are infinitely many points on $L$ that are
at distance less than $\epsilon$ from $c$. Hence, since $a'$ and $b'$ 
are in $S$ and $S$ is essentially convex, there must be one such
point in $S$, and we are done. 
\end{proof}

\ignore{
\begin{lem}\label{lem:convex}
Let $\phi(x_1,\dots,x_k)$ be a standard definition of an essentially convex
semi-algebraic relation $S \subseteq {\mathbb R}^k$ and suppose that $\phi$ does not contain  literals of the form $x \neq y$. Then 
$S$ is convex. 
\end{lem}
\begin{proof}
We show that if $S$ is not convex, then it is also not essentially convex. If $S$ is not convex, there are points $a,b \in S$ such there is a point $c \notin S$
on the line segment between $a$ and $b$. Then there must
be a clause $\psi$ in $\phi$ that is violated by $c$. 
The set defined by $\neg \psi$ is an open set (since it is the intersection of a finite number of sets defined by strict polynomial inequalities). Hence, there is an open ball centered at $c$
that satisfies $\neg \psi$, which shows that $S$ is in fact
not essentially convex.
 \end{proof}
}

\begin{thm}\label{thm:main}
A semi-algebraic relation $S \subseteq \R^n$ is essentially convex if and only if it has a convex Horn definition. 
Moreover, when $S$ is even semi-linear then $S$ has a
semi-linear convex Horn definition.
\end{thm}

\begin{proof}
We have already seen in Proposition~\ref{prop:linear-horn-convex} that
every relation defined by a convex Horn formula is essentially convex. We now show the
more difficult implication of the statement.
Let $\phi$ be a standard definition of $S$. 
The proof is by induction on the dimension $d$ of $S$.
For $d=0$, the set $S$ consists of at most one point, and
the statement is trival.
Otherwise, if $|S| \geq 2$, then by essential convexity, Corollary~\ref{cor:excludes},
and Lemma~\ref{lem:dim} we have that $\dim(S) \geq 1$.

For $d>0$, we will construct two formulas $\phi_1,\phi_2$ 
such that $\phi$ is equivalent to $\phi_1 \wedge \phi_2$.
Thereafter, we will show that $\phi_1$ is equivalent to a conjunction of
convex Horn formulas, and that $\phi_2$ defines a closed convex relation (and consequently is convex Horn).
Since finite conjunctions of convex Horn formulas are also convex Horn,
$\phi$ is then equivalent to a convex Horn formula.

We begin by writing
all clauses of $\phi$ as $\alpha \rightarrow \beta$ where $\alpha$ 
is either `true' or a conjunction of polynomial equalities and $\beta$ is either `false' or a disjunction of inequalities. 
This is always possible since a clause 
\[(p_1 \leq 0 \vee \cdots \vee p_k \leq 0 \vee q_1 \neq 0 \vee \cdots \vee q_m \neq 0)\]
is logically equivalent to
\[(q_1=0 \wedge \cdots \wedge q_m=0) \rightarrow (p_1 \leq 0 \vee \cdots \vee p_k \leq 0).\]
Next, we rewrite all clauses $\alpha \rightarrow \beta$ where $\alpha$ is not 
equal to `true', as $(\alpha \rightarrow (\beta \wedge \phi))$. 
Let $\phi_1$ be the conjunction of all the implications of the type 
$(\alpha \rightarrow (\beta \wedge \phi))$ and $\phi_2$ the conjunction 
of the remaining implications, i.e., those
of the type $({\rm true} \rightarrow \beta)$.
The formula $\phi_1 \wedge \phi_2$ is clearly equivalent to the formula $\phi$.

We begin by studying the formula $\phi_1$.
Let
$\alpha \rightarrow (\beta \wedge \phi)$ be a clause from $\phi_1$, 
let $V$ be the variety defined by $\alpha$, and let $U$ be the set defined
by $\beta \wedge \phi$. Observe that $U \subseteq S$.
We now show that
the intersection of the set $U$ with a semi-algebraic convex set $C \subseteq V$ can be defined by a convex Horn formula. We make two claims about the set $U \cap C$:

\medskip

\noindent
{\em Claim 1.} $U \cap C$ is essentially convex. 
For arbitrary points
 $a,b \in U \cap C$, let $L_{ab}$ denote the line
segment from $a$ to $b$, and $X_{ab}=\{x \in L_{ab} \; | \; x \not\in U \cap C\}$.
 Suppose for contradiction that
there exist $a,b \in U \cap C$ such that $X_{ab}$ is an infinite set.
  Since $a,b \in C$ and $C$ is convex, 
$L_{ab} \subseteq C$ which implies that 
$X_{ab}=\{x \in L_{ab} \; | \; x \not\in U\}$.
  Moreover, $C \subseteq V$ so $X_{ab} \subseteq V$.
Now recall that $a,b \in S$ since $U \cap C \subseteq U \subseteq S$: thus, there are
 infinitely many points (those that are in $X_{ab}$) 
between $a,b \in S$ that are in $V$ but not in $U$. 
This shows that no point in $X_{ab}$ satisfies
$\alpha \rightarrow (\beta \wedge \phi)$, and $X_{ab} \cap S = \emptyset$.
This fact
contradicts the essential convexity of $S$.

\medskip

\noindent
{\em Claim 2.} $U \cap C$ has smaller dimension than $S$. 
Let $T$ be the set $S \setminus (U \cap C)$.
It suffices to show that $U \cap C$
is a subset of $\bar T \setminus T$, because Lemma~\ref{lem:dim}
asserts that $\dim(\bar T \setminus T) < \dim(T) \leq \dim(S)$.

The set $S$ must contain a point $p$ that is not in $V$,
because if $S \subseteq V$ then we could replace 
the clause of $\phi$ that was re-written to 
$\alpha \rightarrow (\beta \wedge \phi)$ 
by $\beta$ and obtain a formula that is equivalent to $\phi$;
this contradicts the assumption 
that $\phi$ is a standard definition of $S$. 

To show that $(U \cap C) \subseteq \bar T \setminus T$, 
let $x$ be an arbitrary point in $U \cap C$.
Only finitely many points on the line segment 
between $p$ and $x$ can be
from $(U \cap C) \subseteq V$, because otherwise 
Proposition~\ref{prop:indep} implies that $V$ must contain the entire
line between $x$ and $p$, including $p$, a contradiction.
Also the set $S$ contains all but finitely many points on the line segment between $p$ and $x$:
this is by essential convexity of $S$, since $x \in U \cap C \subseteq S$ and 
$p \in S$. Hence, we can choose a sequence of points from $T = S \setminus (U \cap C)$ on this line segment
that approaches $x$, which shows that $x \in \bar T$.

\medskip

Since $U \cap C$ is semi-algebraic, essentially convex, and has smaller dimension
than $S$, it follows by the inductive assumption that it can be defined by a convex Horn formula.
Thus, $\phi_1$ is equivalent to a finite conjunction of convex Horn formulas.

We claim that $\phi_2$
defines a closed convex set $D$. 
This follows from Lemma~\ref{lem:closure-convex}, since $D$ is in fact the closure of $S$. 
To see this, observe
that $D$ is clearly a closed set, $D$ contains $S$, and hence $\bar S \subseteq \bar D = D$. To prove that $D \subseteq \bar S$,
let $y$ be from $D \setminus S$. 
Consider the clauses
$\alpha_1 \rightarrow (\beta_1 \wedge \phi),\ldots,\alpha_l \rightarrow (\beta_l \wedge \phi)$ 
of $\phi_1$, and let $V_i$, for $1 \leq i \leq l$, be the variety 
$\{x \in {\mathbb R}^k \; | \; \mbox{$x$ satisfies $\alpha_i$}\}$. 
There must be a point $q$ in $S$ that is not contained in the set $W = \bigcup_{i \leq l} V_i$;
otherwise, Proposition~\ref{prop:indep} implies that there exists an $i \leq l$ such that $S \subseteq V_i$. In other words, all points in $S$ satisfy 
$\alpha_i$. This is in contradiction to the assumption
that $\phi$ is a standard definition of $S$, since in this case the formula
$\phi$ is equivalent to the formula where the clause of $\phi$
that has been rewritten to $\alpha_i \rightarrow (\beta_i \wedge \phi)$
is replaced by $\beta_i$.
Only finitely many points on the line segment $L$ between $q$ and $y$ can be
from $W$, because otherwise Lemma~\ref{lem:algebra} implies that $W$ 
contain the entire
line between $y$ and $q$, including $q$, a contradiction.
Hence, $y \in \bar S$.

Finally, consider the case that $S$ is semi-linear. By Lemma~\ref{lem:standard-def}, we can choose
$\phi$ to be a standard definition which is semi-linear (and only uses parameters in $\mathbb Q$).
Then the proof above leads to a semi-linear convex Horn definition of $S$.
 \end{proof}

\section{Applications}
\subsection{Semi-linear constraint languages}\label{sect:semilinear}

We will now show that a finite semi-linear expansion $\Gamma$ of $\lingamma$ has a polynomial-time tractable constraint satisfaction problem if and only if all relations of $\Gamma$ are essentially convex (unless P = NP).  Recall that a relation is semi-linear if it has a first-order
definition in $({\mathbb R}; +,1,\leq)$.
A quantifier-free first-order formula in CNF is called \emph{Horn-DLR}~\cite{JonssonBaeckstroem} (where `DLR' stands for {\em disjunctive linear relations})
if its clauses are of the form $$p_1 \neq 0 \vee \dots \vee p_k \neq 0$$
or of the form $$p_1 \neq 0 \vee \dots \vee p_k \neq 0 \vee 
p_0 \leq 0$$ 
where $p_0, p_1, \dots, p_k$ are 
linear terms with rational coefficients.
A semi-linear relation is called 
\emph{Horn-DLR} if it can be defined by a Horn-DLR formula.

\begin{thm}[see~\cite{Disj,JonssonBaeckstroem,Koubarakis}]\label{tractability}
Let $\Gamma$ be a structure with domain ${\mathbb R}$ whose relations are Horn-DLR.
Then $\Csp(\Gamma)$ is in P.
\end{thm}
%

In this section, we show the following.

\begin{thm} \label{LPresult}
Let $\Gamma=(\lingamma,S_1,\dots,S_l)$ be a constraint language such that $S_1,\dots,S_l$ are semi-linear relations.
Then, either each relation $S_1,\dots,S_l$ is Horn-DLR
and $\Csp(\Gamma)$ is in P, or $\Csp(\Gamma)$ is
NP-complete.
\end{thm}

In order to prove Theorem~\ref{LPresult}, we need to
characterize convex and essentially convex semi-linear relations.
This is done in Lemma~\ref{lem:polytope} and Theorem~\ref{thm:semi-linear-main},
respectively.

Let $P_1,\ldots,P_n$ be (possibly unbounded) polyhedra defined such that
$P_i=\{x \in {\mathbb R}^k \; | \; A_ix \leq b_i\}$.
Bemporad et al.~\cite{Bemporad} define the {\em envelope} of $P_1,\ldots,P_n$ ($env(P_1,\ldots,P_n)$) 
as the polyhedron

\[\{x \in {\mathbb R}^k \; | \; A_1'x \leq b_1',\ldots,A_n'x \leq b_n'\}\]
where $A_i'x \leq b_i'$ is the subsystem of $A_i x \leq b_i$
obtained by removing all the inequalities not valid for the
other polyhedrons $P_1,\ldots,P_{i-1},P_{i+1},P_n$.  
We note that if $P_1,\ldots,P_n$ are defined with coefficients
from ${\mathbb Q}$, then $env(P_1,\ldots,P_n)$ can be described
by rational coefficients, too.
By combining Theorem~3 with Remark~1 in~\cite{Bemporad}, it follows
that $\bigcup_{i=1}^n P_i$ is convex if and only if
$\bigcup_{i=1}^n P_i=env(P_1,\ldots,P_n)$.

\begin{lem}\label{lem:polytope}
A closed semi-linear relation $S$ is convex if and only if it has a primitive
positive definition in $({\mathbb R}; +,1,\leq)$.
\end{lem}
\begin{proof}
It is straightforward to verify that relations with a primitive positive
definition in $({\mathbb R};+,1,\leq)$ are convex; each relation defines
a convex set and the intersection of convex sets is convex itself.

For the converse, let $\phi=\psi_1 \vee \cdots \vee \psi_m$ be a quantifier-free definition of $S$
over the structure $({\mathbb R};+,-,\leq)$ with parameters from ${\mathbb Q}$,
written in disjunctive normal form.
If there is a disjunct $\psi_i$ that contains a literal $p \neq q$, for linear terms $p$ and $q$, then
split the disjunct into two; one containing $p-q < 0$ and one
containing $p-q > 0$. By repeating this process, every literal
$p \neq q$ can be removed. Similarly, every literal $p=q$ 
can be replaced by $p-q \leq 0 \; \wedge \; q-p \leq 0$.
Thus, we may assume that every literal appearing in the $\psi_i$
is of the type $p \leq 0$ or $p < 0$, for a linear term $p$.
 Let $D_1,\dots,D_m$ be the sets defined by $\psi_1,\dots,\psi_m$, respectively.

Now recall that the topological closure operator preserves
finite unions, i.e., $\bar D_1 \cup \cdots \cup \bar D_m = \overline{D_1 \cup \cdots \cup D_m}$.
Hence, 

\[S=D_1 \cup \cdots \cup D_m \subseteq \bar D_1 \cup \cdots \cup \bar D_m = \overline{D_1 \cup \cdots \cup D_m}=\bar{S}=S \; \]
and $S=\bar D_1 \cup \cdots \cup \bar D_m$.
We now note that if
$P=\{x \in {\mathbb R}^k \; | \; Ax \leq b, \; Cx < d\}$ and $P \neq \emptyset$, then
$\bar P=\{x \in {\mathbb R}^k \; | \; Ax \leq b, \; Cx \leq d\}$,
cf. Case (i) of Proposition 1.1 in~\cite{Goberna:etal}.
Thus, each
$\bar D_i$ equals $\{x \in {\mathbb R}^k\; | \; A_i x \leq b_i\}$ for some rational $A_i,b_i$.
Furthermore, $S=\bigcup_{i=1}^m \bar D_i$ is convex so
$S=env(\bar D_1,\ldots,\bar D_m)$. 
This implies that
$S=\{x \in {\mathbb R}^k \; | \; C_1 x \leq d_1,\ldots,C_m x \leq d_m\}$ for some rational matrices $C_1,\ldots,C_m$ and rational vectors $d_1,\ldots,d_m$.
It is easy to see that each $C_i x \leq d_i$ is pp-definable in $({\mathbb R}; +,1,\leq)$ by the same technique as in the proof of Lemma~\ref{lem:pp-rational}, and
this concludes the proof.
 \end{proof}

\begin{thm}\label{thm:semi-linear-main}
A semi-linear relation $S$ is essentially convex if and only if $S$ is Horn-DLR.
\end{thm}
\begin{proof}
We first prove that every Horn-DLR relation $S$ is essentially convex.
Let $\phi$ be a Horn-DLR definition of $S$.
Suppose for contradiction that there are $a,b \in S$ and an infinite set $I$ of points on
the line segment $L$ between $a$ and $b$ is not contained in $S$. Since $\phi$ has finitely many conjuncts,
there is a conjunct $\psi$ in $\phi$ that is false for an infinite subset $I'$ of $I$. 
If $\psi$ is of the form  $p_1 \neq 0 \vee \dots \vee p_k \neq 0$,
 then all points in $I'$ satisfy $p_1 = \dots = p_k = 0$.
 By Lemma~\ref{lem:algebra}, the entire line $L$ satisfies $p_1 = \dots = p_k = 0$.
This contradicts the assumption that $a \in L$ and $b \in L$ satisfy $\phi$. 
If $\psi$ is of the form $p_1 \neq 0 \vee \dots \vee p_k \neq 0 \vee p_0 \leq 0$,
then all points in $I'$ satisfy $p_1 = \dots = p_k = 0$ and $p_0 >0$. 
Again by Lemma~\ref{lem:algebra} we find that both $a$ and $b$ must satisfy $p_1 = \dots = p_k = 0$.
Since $a,b$ also satisfy $\psi$, we conclude that both points satisfy $p_0 \leq 0$. But then
also all points in $L$ must satisfy $p_0 \leq 0$, which contradicts the fact that the points in $I'$ satisfy $p_0>0$.

The other direction of the statement can be derived from 
Theorem~\ref{thm:main} as follows. 
Let $S$ be an essentially convex semi-linear relation. 
By Theorem~\ref{thm:main}, $S$ has a semi-linear convex Horn definition $\psi$.
We prove by induction on the level of $\psi$ 
that $\psi$ is Horn-DLR.
If the level is 0, then $S$ is closed and convex and the claim
follows from Lemma~\ref{lem:polytope}.
Now suppose that $\psi$ has level $i>0$ and is of the form
$p_1 \neq 0 \vee \dots \vee p_k \neq 0 \vee \psi'$,
where $p_1 = \dots = p_k = 0$ defines a set $V$,
and $\psi'$ defines a set $U$ such that for every semi-algebraic convex set $C \subseteq V$ the set $C \cap U$ has a convex Horn definition of level strictly smaller than $i$.
Since $\psi$ is semi-linear, the terms $p_1,\dots,p_k$ are 
linear.
Hence, $V$ is convex, and by taking $C:=V$ in the statement above
we see that $V \cap U$ has a convex Horn definition of level strictly
smaller than $i$.
By the inductive assumption, $V \cap U$ has a definition by a Horn-DLR formula $\phi$ with 
clauses $\phi_1, \dots, \phi_m$.
Then $\phi' = \bigwedge_{1 \leq i \leq m} 
(p_1 \neq 0 \vee \dots \vee p_k \neq 0 \vee \phi_i)$ is clearly Horn-DLR.
We claim that $\phi'$ defines $S$.
First suppose that $a \in \neg V$.
In this case, $a$ clearly satisfies $\phi'$ and this
is justified by the fact $\neg V \subseteq S$.
Suppose instead that $a \in V$. Then $a$ satisfies $\phi'$ if and only if it
satisfies $\phi$, and since $\phi$ defines $V \cap U$ this is the case 
if and only if $a$ satisfies $\psi'$.

Finally, the statement holds if $\psi$ is the conjunction of finitely
many convex Horn formulas 
(which are Horn-DLR by inductive assumption).
 \end{proof}

\begin{proof}[Proof of Theorem~\ref{LPresult}]
If all relations of $\Gamma$ are Horn-DLR, then $\Csp(\Gamma)$ can be solved in polynomial time (Theorem~\ref{tractability}).
Otherwise, if there is a relation $S$ from $\Gamma$ that is not Horn-DLR, then Theorem~\ref{thm:semi-linear-main} shows that $S$ is not essentially convex, and NP-hardness of $\Csp(\Gamma)$ follows by Lemma~\ref{semilinhard}.

So we only have to show that $\Csp(\Gamma)$ is in NP.
Let $\Phi$ be an arbitrary instance of $\Csp(\Gamma)$.
By Theorem~\ref{thm:LP-QE} every relation of $\Gamma$
has a quantifier-free definition in conjunctive normal form
over $(\mathbb R;+,-,\leq)$ with rational parameters. One can now non-deterministically guess \emph{one} literal from each clause of in the defining formula for each constraint and verify -- 
in polynomial-time by Theorem~\ref{tractability} --
that all the selected literals are simultaneously satisfiable.
 \end{proof}

\subsection{Generalized linear programming}
\label{sect:glp}
In this section, we study generalizations of the following problem.

\cproblem{Linear Programming (LP)}
{A finite set of variables $V$, a vector $c \in {\mathbb Q}^{|V|}$, a number $M \in \mathbb Q$, and a finite set of linear inequalities of the form $a_1x_1+\dots+a_nx_n \leq a_0$ where $x_1,\dots,x_n \in V$ and 
$a_0,\dots,a_n \in \mathbb Q$. All rationals are given by numerators and denominators represented in binary.}
{Is there a vector $x \in {\mathbb R}^{|V|}$ that satisfies the inequalities
and $c^Tx \geq M$?}

We generalize LP as follows. 
Let $\Gamma$ be a structure $(\Gamma_{\rm lin}, R_1,\ldots,R_m)$
such that $R_1,\ldots,R_m$ are semi-linear relations.

\cproblem{Generalized Linear Programming for $\Gamma$ (GLP$(\Gamma)$)}
{A finite set of variables $V$, a vector $c \in {\mathbb Q}^{|V|}$, a number $M \in \mathbb Q$, and a finite set $\Phi$ 
of expressions of the form $R(x_1,\dots,x_k)$ where $R$ is a relation from $\Gamma$ and $x_1,\ldots,x_k \in V$.}
{Is there a vector $x \in {\mathbb R}^{|V|}$ that satisfies the constraints
and $c^Tx \geq M$?}

This can indeed be viewed as a generalization of LP because of Proposition~\ref{prop:lp-pp}: the problem LP is polynomial-time equivalent to the problem
$\GLP(\lingamma)$.

\begin{thm}
Let $\Gamma=({\mathbb R}; \lingamma,R_1,\ldots,R_l)$ be a structure with semi-linear relations $R_1,\dots,R_l$. Then, either each $R_i$ is Horn-DLR and $\GLP(\Gamma)$ is in P, or $\GLP(\Gamma)$ is NP-hard.
\end{thm}
\begin{proof}
If there is an $R_i$ that is not Horn-DLR, then the relation
$R_i$ is not essentially convex by Theorem~\ref{thm:semi-linear-main}, and
$\Csp((\lingamma,R_i))$ is NP-hard by Theorem~\ref{LPresult}. Clearly, 
$\GLP(\Gamma)$ is NP-hard, too.

Assume instead that
each $R_i$ is Horn-DLR. We present an algorithm that actually solves a more general problem that includes $\GLP(\Gamma)$.
 Let $\Phi$ be an arbitrary satisfiable Horn-DLR formula\footnote{Note that if we are given an instance of
$\Csp((\lingamma,R_1,\ldots,R_l))$, then it can 
be transformed into an equivalent Horn-DLR formula in polynomial time since
there is only a finite number of relations in the given structure.
Hence, there is no loss of generality in considering Horn-DLR formulas instead
of CSP instances.
Also note that the resulting formula is (up to a multiplicative constant
depending on the structure)  of the same size as the CSP instance.} over
variable vector $\bar{x}=(x_1,\ldots,x_n)$ and let $c$ be a rational $n$-vector.

We assume additionally that $\Phi \wedge D$ is satisfiable for every
disequality literal (i.e., literal $p(\bar{x}) \neq a$) 
$D$ appearing in $\Phi$. If $\Phi \wedge D$
is not satisfiable, then every occurrence of $D$
in $\Phi$ can be removed without changing the set defined by the formula.
Furthermore, this check can be carried out in polynomial time
by Theorem~\ref{tractability}. Hence, we may assume that $\Phi$ has
this additional property (which we will refer to as $(*)$) without loss
of generality.
Let $\Phi=\Phi' \wedge \Phi''$ where $\Phi'$ consists of the
clauses not containing any disequality literal $p(\bar x) \neq a$, 
and $\Phi''$ consists of the remaining clauses.

Given $\Phi$, our algorithm returns one of the following three answers:

\begin{iteMize}{$\bullet$}
\item
`unbounded': for every $K \in {\mathbb Q}$, there exists
a solution $y$ such that $c^Ty \geq K$;

\item
`optimum: $K$': there exists a $K \in {\mathbb Q}$ and a solution $y$ such that
$c^Ty=K$, but there is no solution $y'$ such that $c^Ty' > K$;

\item
`optimum is arbitrarily close to $K$': there exists a $K \in {\mathbb Q}$
such that
there is no solution $y$ satisfying $c^Ty \geq K$, but
for every $\epsilon > 0$ there is a solution $y'$
with $c^Ty' \geq K-\epsilon$.
\end{iteMize}

\noindent
We claim that the following algorithm solves the task described above in polynomial time.

\medskip

\noindent
{\bf Step 1.}
Maximize $c^T \bar x$ over $\Phi'$ (by using some polynomial-time algorithm for
linear programming).
Let $K$ denote the optimum.
If $K=\infty$, then return `unbounded' and stop.

\medskip

\noindent
{\bf Step 2.}
Check whether $\Phi \wedge c^T \bar x=K$ is satisfiable. Note that $c^T \bar x=K$ has a primitive positive definition
in $\lingamma$, which furthermore can be computed in polynomial time by Lemma~\ref{lem:pp-rational}.
Therefore this check can be reduced to deciding satisfiability of Horn-DLRs.
If $\Phi \wedge c^T \bar x=K$ is satisfiable, then return `optimum: $K$'.
If this is not the case, then return `optimum is arbitrarily close
to $K$'.

\bigskip

We first show that the algorithm runs in polynomial time.
{\bf Step 1} takes polynomial time since maximizing $c^T \bar x$ over $\Phi'$
is equivalent to solving a linear program with size polynomially bounded
in the size of $\Phi$.
Finally, {\bf Step 2} takes polynomial time
due to Lemma~\ref{lem:pp-rational} and Theorem~\ref{tractability}.

Next, we prove the correctness of the algorithm.
Correctness is obvious if the algorithm answers `optimum: $K$' in {\bf Step 2}.
For the remaining cases, we need to make a couple of observations.
Define $S=\{x \in {\mathbb R}^n \; | \; \mbox{$x$ satisfies $\Phi$}\}$ and
$S'=\{x \in {\mathbb R}^n \; | \; \mbox{$x$ satisfies $\Phi'$}\}$. 
Let $D_1,\ldots,D_m$ denote the disequality literals appearing
in $\Phi$.
Let $H_i$ be the set $\{x \in {\mathbb R}^n \; | \; \mbox{$x$ does not satisfy $D_i$}\}$, $1 \leq i \leq m$, and note that each $H_i$ is a hyperplane.

\medskip

\noindent
{\em Observation 1.} The formula 
$\Phi^- \equiv \Phi \wedge D_1 \wedge \cdots \wedge D_m$ is satisfiable. 

Otherwise, 
$S \subseteq H_1 \cup \cdots \cup H_m$. The set $S$ is
essentially convex and each $H_m$ is a variety,
so there exists an $1 \leq i \leq m$ such that $S \subseteq H_i$
by Proposition~\ref{prop:indep}. 
Consequently, $\Phi \wedge D_i$ is 
not satisfiable which contradicts the fact that $\Phi$ has property $(*)$.

\medskip

\noindent
{\em Observation 2.} For every $\epsilon > 0$ there is a $y \in S$
satisfying $|c^Tw - c^Ty| < \epsilon$.
 
Let $d(\cdot,\cdot)$
denote the Euclidean distance in ${\mathbb R}^n$, i.e.,
$d(a,b)=\sqrt{\sum_{i=1}^n (a_i-b_i)^2}$, and $|| \cdot ||$ the
corresponding norm, i.e., $||a||=\sqrt{a^Ta}$.

Arbitrarily choose a point $z$ that satisfies $\Phi^-$; 
this is always possible by Observation 1. Consider the line
segment $L$ between $z$ and $w$.
Note the following: if $H$ is a hyperplane in ${\mathbb R}^n$, then 
either $H$ intersects $L$ in at most one point or $L \subseteq H$.
Also note that $w,z \in S'$ and
$S'$ is convex so $L \subseteq S'$.
Arbitrarily choose a clause $C \in \Phi''$ and
assume $C = (p_1(\bar{x}) \neq 0 \vee \cdots \vee 
p_k(\bar{x}) \neq 0 \vee p_0(\bar{x}) \leq 0)$.
Assume that there exist two distinct points $a,b \in L$
such that $p_1(a)=p_1(b)=0$. 
If so, then every point $c \in L$ satisfies
$p_1(c)=0$. This is not possible since $z \in L$ satisfies
$D_1 \wedge \dots \wedge D_m$, and in particular
$p_1(z) \neq 0$. Hence, at most one point $c \in L$ 
satisfies $p(c)=0$, and $c$ is the only point in $L$
that potentially does not satisfy the clause $C$.
This implies that only finitely many points in $L$ do not
satisfy $\Phi''$, and it follows
that for every $\delta > 0$ there is
a point $y \in S$ such that $d(w,y) < \delta$.


We proceed by showing that if
$w,y \in {\mathbb R}^n$ and $d(w,y)=d$, 
then $|c^Tw-c^Ty| \leq ||c|| \cdot d$.
This follows from the Cauchy-Schwarz inequality (that is, 
$|a^Tb| \leq ||a|| \cdot ||b||$ for vectors $a,b$ in ${\mathbb R}^n$):
\begin{align*}
|c^Tw-c^Ty| & =  |c^T(w-y)| \leq
||c|| \cdot ||w-y|| = ||c|| \cdot d(w,y) = ||c|| \cdot d.
\end{align*}
To find a vector $y$ that satisfies $|c^Tw - c^Ty| < \epsilon$,
we simply choose $y \in S$ such that
$d(w,y) < \frac{\epsilon}{||c||}$; we know that such a $y$ exists by the
argument above.

\medskip

If the algorithm outputs `unbounded' in {\bf Step 2}, then
arbitrarily choose a sufficiently large number $k$ and note that there exists a vector
$w \in S'$ such that $c^Tw \geq k$.
 By Observation 2, there exists
a vector $y \in S$ such that $|c^Tw-c^Ty| < 1$. Hence, $S$
has unbounded solutions, too.

Assume finally that the algorithm answers `optimum is arbitrarily close to $K$'
in {\bf Step~3};
Observation 2 immediately proves correctness in this case.
\end{proof}

\section{Open Problems}
\label{sect:concl}
The most prominent open question is whether there are
there are essentially convex relations $S$ with a first-order definition
in $(\mathbb R;*,+)$
such that $\Csp((\lingamma,S))$
is NP-hard.
Resolving this question is probably difficult, since the following
closely related problem is of unknown computational complexity:

\cproblem{Feasibility of Convex Polynomial Inequalities}
{A set of variables $V$, a set of polynomial inequalities each of which defining a convex set; the coefficients of the polynomials are rational numbers where
the numerators and denominators are represented in binary.}
{Is there a point in ${\mathbb R}^{|V|}$ that satisfies all inequalities?}

One may note that the problems we have considered could be easier since
they are defined over finite constraint languages.
On the other hand, convexity is much more restrictive than essential convexity; moreover, we are only
given polynomial inequalities (there are convex semi-algebraic relations that cannot be defined as the intersection of convex polynomial inequalities). Still, the computational complexity
of the Feasibility Problem of Convex Polynomial Inequalities is open.

\emph{Convex} semi-algebraic relations are of particular interest
in the quest for efficiently solvable semi-algebraic constraint languages because of a conjectured link to semidefinite programming.
Every semidefinite representable
set is convex and semi-algebraic. Recently, Helton, Vinnikov and Nie
showed that
the converse statement is true in surprisingly many cases
and conjectured that it remains true in general~\cite{HeltonNie}.


\section*{Acknowledgements}
We would like to thank Johan Thapper for pointing out inaccuracies in an earlier version
of the article, Frank-Olaf Schreyer for helpful discussions in the
early stages of this work, and the referees for their detailed remarks.
\bibliographystyle{abbrv}
\bibliography{../../global}

\end{document}